\documentclass[%
reprint%
,aip
,a4paper,superscriptaddress,nofootinbib,showkeys]%
{revtex4-1}
\usepackage{stdpreamble}
\usepackage{enumitem}

\newcommand{\bra}{(}
\newcommand{\ket}{)}
\newcommand{\im}{\text{im}\,}
\newcommand{\Z}{\mathbb{Z}}
\newcommand{\ST}{ST}
\newcommand{\HT}{HT}
\newcommand{\torgen}{\lambda_\varphi}
\newcommand{\polgen}{\lambda_\vartheta}
\newcommand{\radgen}{\lambda_r}
\newcommand{\toroform}{\sigma_\varphi}
\newcommand{\poloform}{\sigma_\vartheta}
\newcommand{\radform}{\sigma_r}
\newcommand{\torflux}{\Psi_T}
\newcommand{\polflux}{\Psi_P}
\newcommand{\torloop}{L_T}
\newcommand{\polloop}{L_P}
\newcommand{\radflux}{\Psi_R}
\newcommand{\torpath}{C_T}
\newcommand{\polpath}{C_P}
\newcommand{\radpath}{C_R}
\newcommand{\polocut}{S_P}
\newcommand{\toribbon}{S_T}
\newcommand{\torfield}{\nu_\varphi}
\newcommand{\polfield}{\nu_\vartheta}
\newcommand{\radfield}{\tilde{\nu}_\rho}
\newcommand{\gradphi}{"d\varphi"}
\newcommand{\Harm}{\text{Harm}}

\newcommand{\nochange}{%
}
\begin{document}
\title{Gauge freedom in magnetostatics and the effect on helicity in toroidal volumes}
\author{David Pfefferlé}
\author{Lyle Noakes}
\affiliation{The University of Western Australia, 35 Stirling Highway,
  Crawley WA 6009, Australia}

\date{\today}

\begin{abstract}
Magnetostatics defines a class of boundary value problems in which the topology of the domain plays a subtle role. For example, representability of a divergence-free field as the curl of a vector potential comes about because of homological considerations. With this in mind, we study gauge-freedom in magnetostatics and its effect on the comparison between magnetic configurations through key quantities such as the magnetic helicity. For this, we apply the Hodge decomposition of $k$-forms on compact orientable Riemaniann manifolds with smooth boundary, as well as de Rham cohomology, to the representation of magnetic fields through potential $1$-forms in toroidal volumes. An advantage of the homological approach is the recovery of classical results without explicit coordinates and assumptions about the fields on the exterior of the domain. In particular, a detailed construction of a \emph{minimal gauge} and a formal proof of \emph{relative helicity formulae} are presented.
\end{abstract}


\maketitle

\section{Introduction}
\label{sec:intro}

{\nochange

Explicit calculations\cite{vmec,spec} of the magnetic field are often made in terms of the vector potential $\bm{A}(\bm{x})$; the magnetic field is then found by applying the \emph{curl} operator, i.e. $\bm{B}(\bm{x}) := \nabla\times \bm{A}$. The resulting magnetic field is then automatically divergence-free, $\nabla\cdot\bm{B}=\nabla\cdot(\nabla\times\bm{A})=0$, thereby satisfying one of Maxwell's equations. However the representation of the field by a vector potential is not unique, and this has consequences for key quantities such as the total magnetic helicity.
For example, let $\bm{A}_1$ and $\bm{A}_2$ be two vector potentials for a given magnetic configuration. They describe the same magnetic field precisely when their difference $\bm{s}=\bm{A}_1-\bm{A}_2$ is curl-free. Indeed, the part of the vector potential in the kernel of the curl operator is superfluous to the physics. The freedom to pick $\bm{A}$ among all vector fields that produce the same curl is called \emph{gauge freedom}. One could decide before performing any calculation on a convention or a \emph{gauge}, so that when the magnetic fields are the same, the vector potentials would match point-wise. Such a convention is called \emph{gauge fixing}.

A simpler occurrence of gauge freedom is the representation of a conservative force as the gradient of a potential function $V(\bm{x})$, namely $\bm{F}(\bm{x}):= -\nabla V$. Here similarly, different potentials represent the same force precisely when they differ by some constant.  The freedom is minimal in the sense that one has to agree upon just one number. To fix the gauge, one could adopt the convention that all potential functions should vanish at a point of reference, or that the average of all potential functions should vanish, or that the minimum of all potential functions should be zero. The convention affects, for example, the definition of the system's energy, but energy differences can compared in a gauge-invariant way.


Gauge freedom for vector potentials in magnetostatics is more involved than for conservative forces. While the physics (Maxwell equations) is gauge-independent, some important quantities depend on the convention that is chosen and on the domain on which the magnetic fields are represented: a convention that works well in a closed ball may not be appropriate for a solid torus.
}

Understanding gauge freedom is particularly important for the computation and interpretation of total magnetic helicity, which is the volume integral of the gauge-dependent quantity $\bm{A}\cdot \bm{B}$. It is known that total magnetic helicity is related to topological properties of the magnetic field~\citep{moreau-1961,moffatt-1969}. On a simply-connected orientable three-dimensional Riemannian manifold with boundary, total helicity is the asymptotic Hopf invariant of the magnetic field, measuring the average linking of its field-lines~\cite{arnold-1974}. The justification~\citep{berger-field-1984,berger-1999} uses the Biot-Savart law to express the vector potential as the integral of the magnetic field over a suitable kernel, and then identifies the magnetic helicity, in the form of a double-integral, as the average Gauss linking number of field-lines. The problems with this are, first, that the identification works if and only if the magnetic field is tangential to the boundary (otherwise the outside field must also be taken into account), and second, that the vector potential resulting from Biot-Savart is always divergence-free, which means that the so-called \emph{Coulomb} or \emph{minimal} gauge must be assumed.

Gauge-invariant or \emph{relative helicity} formulae have been derived for the case where the magnetic field remains tangential to the boundary~\cite{berger-field-1984,jensen-chu-1984,finn-antonsen-1985}. The way these formulae are justified is by asserting regularity of the solutions ``at infinity'', by performing clever book-keeping of the self and mutual helicities of magnetic fields restricted to simply-connected volumes of space, and finally by approximating the mutual helicity of the vacuum components as a product of fluxes. Based on intuition, the self-helicity of the vacuum field is for example set to zero~\cite{berger-1999}. The construction does not seem to extend to general boundary conditions on $\bm{B}$ and less regular fields. These formulae are widely used in the context of helicity injection in fusion plasmas~\cite{jensen-chu-1984} and reconnection processes in solar flare dynamics~\cite{taylor-1986}.

This paper explores an approach based on \emph{homology} to address the representation of magnetic fields via vector potentials, treating this exercise as a boundary value problem on a compact orientable Riemannian manifold. Using tools of differential geometry and algebraic topology from \emph{de Rham cohomology}~\citep{derham-1931} and \emph{Hodge theory}~\citep{hodge-1941}, a justification of the relative helicity formulae is provided.

The paper is organised as follows. Appendix \ref{sec:hodge-theory} reviews the general elements of de Rham cohomology and the Hodge decomposition of $k$-forms on compact manifolds. In section \ref{sec:cohomology}, we study the cohomology of solid toroidal and hollow toroidal volumes and write explicit generators for the de Rham cohomology groups. In section \ref{sec:hodge-decomp}, we list the conditions under which the magnetic field can be represented unambiguously by a vector potential, and how to decompose the latter in terms of a physical component and gauge terms. A pair of minimal gauges are proposed, depending on the topology of the toroidal domains (solid or hollow). In section \ref{sec:helicity}, the notion of magnetic helicity is discussed as well as how gauge freedom affects its computation. In the case of perfectly conducting boundary conditions, we show how to define for the two toroidal domains a relative total magnetic helicity, in such as way that its value is gauge-invariant. This enables the total magnetic helicity of two different magnetic field representations to be compared in a meaningful way.

It is assumed that the reader is somewhat familiar with notions of differential geometry and exterior calculus. It is recommended to read appendix~\ref{sec:hodge-theory} first to acclimate to the notation. Some readers may skip to section \ref{sec:helicity}, which highlights the practical consequences for the computation of total magnetic helicity.

\section{Cohomology of toroidal volumes}
\label{sec:cohomology}
The results of Hodge theory and de Rahm cohomology, as reviewed in appendix \ref{sec:hodge-theory}, are applied to the magnetostatics problem of representing a magnetic field through a vector potential on a toroidal volume, such as in a tokamak, a stellarator or even a knotted configuration, e.g. flux tubes in a reversed-field pinch or in the solar corona. 

Let us consider the magnetic field $\bm{B}:=(\star B)^\sharp$ to be the \emph{sharp} (isomorphism between vector fields and $1$-forms induced by the Riemannian metric) of the Hodge dual of a $2$-form called the \emph{magnetic flux} $B\in W^1\Omega^2(M)$ on $M$, a compact orientable Riemannian manifold with boundary. Knowing that $B$ is closed, $dB=0$ (corresponding to $\nabla\cdot \bm{B}=0$), we would like to know when it is exact, namely if the relation $B=dA$ between the magnetic flux and the potential $1$-form $A\in W^1\Omega^1(M)$ holds. The latter conveys the usual identity between the vector potential $\bm{A}:=A^\sharp$ and the magnetic field $\bm{B}=\curl\bm{A}$. To make proper use of the Hodge-Friedrichs-Morrey decomposition theorem (HFMD, see section~\ref{sec:hfmd}), we first characterise the de Rham cohomology groups, enabling us to identify the spaces of \emph{harmonic fields} (see Theorem \ref{thm:hodge-isomorphism}).

{\nochange
The first manifold $M=\ST \subset \R^3$ under consideration is the image of an embedding in $\R^3$ of the solid torus $D^2\times S^1=\{(x_1,x_2,x_3,x_4)\in \R^4|\, (x_1,x_2)\in D^2, (x_3,x_4)\in S^1\}\subset\R^4$, i.e. the Cartesian product of the unit disk $D^2=\{(x_1,x_2)\in \R^2|\, x_1^2+x_2^2\leq 1\}\subset \R^2$ and the unit circle $S^1=\{(y_1,y_2)\in \R^2|\, y_1^2+y_2^2=1\}\subset \R^2$. Let $\Phi : D^2\times S^1\to \ST \subset \R^3$ denote the embedding.  The boundary of the manifold $\partial \ST = \Phi(\partial D^2\times S^1)\cong S^1\times S^1 = T^2\subset \R^4$ is diffeomorphic to a torus.

Another manifold of interest is the hollow torus. Let $M=\HT\subset \R^3$ be the image of an embedding of the Cartesian product of a closed annulus and the unit circle. The closed annulus is treated as the product of an interval and another circle. We thus consider the embedding $\Phi:[0,1]\times S^1 \times S^1=\{(x_1,x_2,x_3,x_4,x_5)\in \R^5|\, x_1\in [0,1], x_2^2+x_3^2=1, x_4^2+x_5^2=1\}\to HT\subset \R^3$. The boundary of $\HT$ is the disjoint union of two tori, $\partial \HT\cong T^2_0 \cup T^2_1$ the interior and the exterior.

As subdomains of $\R^3$, $\ST$ and $\HT$ are given Riemannian metrics corresponding to the first fundamental forms from the Euclidean inner product.
}
It will be assumed throughout the paper that the embedding $\Phi$ preserves the orientation of the manifolds, in the sense that the standard volume form $dX\wedge dY\wedge dZ \in \Omega^3(M\subset \R^3)$ pulls back to a positive top form on the pre-image. This condition can be relaxed with the price of an extra sign in several formulae here below.

\subsection{Cohomology of the solid torus in $\R^4$}
{\nochange 
The cohomology of $\ST$ is isomorphic to that of $D^2\times S^1$. The unit disk (poloidal) being contractible, the absolute cohomology of $ST$ is isomorphic to that of the unit circle (toroidal). We list
\begin{align}
  \begin{split}
    H^0(\ST,\R)&\cong H^0(S^1,\R)\cong \R,\\
    H^1(\ST,\R)&\cong H^1(S^1,\R)\cong \R,
  \end{split}
\end{align}
and the rest is trivial $H^2(\ST,\R)=\{0\}$ and $H^3(\ST,\R)= \{0\}$. 
}

The relative cohomology of $ST$ can be deduced either via the Künneth formula or by invoking Poincaré-Lefschetz duality (see Corollary~\ref{thm:poincare-duality}), but the following reasoning provides an intuitive justification. Every point inside $\ST$ is homologous to a point on the boundary so that the zeroth relative cohomology group is trivial. Every closed loop can be moved to the boundary (including the generator of the circle's first homology) so that the first relative cohomology group is also trivial. The second relative homology of $\ST$ consists of surfaces, for which points on $\partial ST$ are discarded, that are closed (cycles) but do not bound any volume. A \emph{poloidal cut}, which opens the solid torus into a cylinder, fits exactly this description; the boundary of the poloidal cut lies on the boundary of the solid torus so that this relative chain is closed. The poloidal cut actually generates the solid torus' second homology (all other relative cycles are homologous to multiples of the poloidal cut), hence, the second relative cohomology is one-dimensional. To deduce $H^3_{dR}(\ST,\partial \ST,d)\cong \Harm_D^3(\ST)$, we recall Theorems~\ref{thm:hodge-isomorphism} and~\ref{thm:hodge-isomorphism2} from appendix~\ref{sec:hodge-theory} that all top forms are closed and Dirichlet but only the "constant" volume form is co-closed (no matter the Riemannian metric), so that $\Harm_D^3(\ST)\cong \{ c\mu \in \Omega^n(M)|\, c\in \R\}$. We thus list
\begin{align}
  \begin{split}
  H^0(\ST,\partial \ST, \R)& = \{0\},\\
  H^1(\ST,\partial \ST,\R)&= \{0\} ,\\
  H^2(\ST,\partial \ST, \R)&\cong \R,\\
  H^3(\ST,\partial \ST, \R)&\cong \R.
\end{split}
\end{align}


Let $\torpath:(0,2\pi)\to D^2\times S^1$, $\varphi\mapsto \torpath(\varphi) = (0,0,\cos\varphi,\sin\varphi)$ be a closed \emph{toroidal path} (our preferred representative of first homology generating class). An explicit representative $\torgen$ of the dual basis class $[\torgen] \in H_{dR}^1(D^2\times S^1,d)$ that generates the first cohomology group is, for example,
 \begin{align}
 \label{eq:rep_phi}
 \torgen :=\frac{1}{2\pi}(-x_4dx_3+x_3dx_4).
\end{align}
Indeed, $d\torgen = 0$ and the integral along any curve homologous to $\torpath$, is
\begin{align*}
 \ll [\torpath],[\torgen]\gg &= \int_{\torpath} \torgen
 = \int_0^{2\pi} \torpath^*\torgen
  = \int_0^{2\pi} \torgen(\torpath') d\varphi\nonumber\\
  &= \frac{1}{2\pi}\int_0^{2\pi}(\sin^2\varphi+\cos^2\varphi)d\varphi
  = 1.
\end{align*}

Let $\polocut:D^2\setminus \partial D^2 \to D^2\times S^1$, $(x_1,x_2)\mapsto \polocut(x_1,x_2) = (x_1,x_2,1,0)$ be a \emph{poloidal cut} (closed relative to the boundary, our preferred generator of second relative homology). An explicit representative $\poloform$ of the dual basis class $[\poloform]\in H^2_{dR}(D^2\times S^1,\partial D^2\times S^1,d)$ that generates the second relative cohomology group is, e.g.
\begin{align}
  \poloform:= \frac{1}{\pi}dx_1\wedge dx_2.
\end{align}
Indeed, $d\poloform=0$ and $\bm{t}\poloform = 0$ and the integral over any surface homologous to $\polocut$ is
\begin{align*}
 \ll[\polocut],[\poloform]\gg= \int_{\polocut} \poloform
  =  \int_{D^2} \polocut^*\poloform
  = \frac{1}{\pi}\int_{D^2} dx_1 dx_2 = 1.
\end{align*}

The wedge product of $\torgen$ and $\poloform$ is a closed Dirichlet $3$-form, i.e.  $d(\poloform\wedge \torgen)=0$ and $\bm{t} (\poloform\wedge \torgen)=\cancel{\bm{t}\poloform}\wedge\bm{t}\torgen=0$. It is thus the representative of a third relative cohomology class. Its integral over the solid torus is computed using the chart $V:D^2\setminus \partial D^2\times (0,2\pi)\to D^2\times S^1$ given by $V(x_1,x_2,\varphi)=(x_1,x_2,\cos\varphi,\sin\varphi)$ (our preferred generator of third relative homology). Explicitly,
\begin{align}
  \label{eq:wedge-canonical}
  \int\limits_{D^2\times S^1} \poloform\wedge \torgen
  &= \int\limits_V\poloform\wedge \torgen
  = \int\limits_{D^2}\int_0^{2\pi} V^*(\poloform\wedge\torgen)\nonumber\\
  &=\frac{1}{2\pi^2}\int\limits_{D^2}dx_1dx_2\int_0^{2\pi} d\varphi = 1.
\end{align}
Hence, $\poloform\wedge\torgen$ is actually a representative of the dual basis class $[\poloform\wedge\torgen]\in H^3_{dR}(D^2\times S^1,\partial D^2\times S^1,d)$ that generates the third relative cohomology group.

\subsection{Cohomology of the hollow torus in $\R^5$}
{\nochange The absolute cohomology of $HT$ is isomorphic to that of the torus $T^2=S^1\times S^1$. We list
\begin{align}
  \begin{split}
    H^0(\HT,\R)&\cong H^0(T^2,\R)\cong \R,\\
    H^1(\HT,\R)&\cong H^1(T^2,\R)\cong \R^2,\\
    H^2(\HT,\R)&\cong H^2(T^2,\R)\cong \R,\\
    H^3(\HT,\R)&\cong \{0\}
  \end{split}
\end{align}
}
By Poincaré-Lefschetz duality (see Corollary~\ref{thm:poincare-duality}), the relative cohomology is
\begin{align}
  \begin{split}
    H^0(\HT,\partial \HT,\R)&\cong \{0\},\\
    H^1(\HT,\partial \HT,\R)&\cong \R,\\
    H^2(\HT,\partial \HT,\R)&\cong \R^2,\\
    H^3(\HT,\partial \HT,\R)&\cong \R.
  \end{split}
\end{align}

Let $\torpath:(0,2\pi)\to [0,1]\times S^1\times S^1$, $\varphi\mapsto\torpath(\varphi)=(0,1,0,\cos\varphi,\sin\varphi)$ be a closed \emph{toroidal path} and $\polpath:(0,2\pi)\to [0,1]\times S^1\times S^1$, $\vartheta\mapsto \polpath(\vartheta)=(0,\cos\vartheta,\sin\vartheta,1,0)$ be a closed \emph{poloidal path} (our preferred generators of first homology). {\nochange Explicit representatives of the dual basis for the first absolute cohomology group are, for example,
\begin{align}
  \begin{split}
  \polgen&:= \frac{1}{2\pi}(-x_3dx_2+x_2dx_3),\\
  \torgen &:= \frac{1}{2\pi}(-x_5dx_4+x_4dx_5),
\end{split}
\end{align}
with the property that $\int_{\torpath} \torgen = 1$, $\int_{C_P}\polgen=1$, $\int_{C_P} \torgen=0$ and $\int_{\torpath}\polgen=0$. A representative for the dual basis class of the second absolute cohomology group is, for example, the wedge product
\begin{align}
  \radform: = \polgen\wedge\torgen,
\end{align}
with the property that $\int_{T^2} \radform=1$ where $T^2=\{.\}\times S^1\times S^1$ is a torus (our preferred generator of second homology).

Let $\radpath:[0,1]\to [0,1]\times S^1\times S^1$, $x\mapsto \radpath(x)=(x,1,0,1,0)$ be a \emph{radial path}. It is closed relative to the boundary, and thus represents a generator of first relative homology. A representative for the dual basis class of the first relative cohomology group is
\begin{align}
  \radgen: =dx_1,
\end{align}
with the property that $\int_{\radpath} \radgen = 1$.
}

Let $\polocut:(0,1)\times (0,2\pi)\to [0,1]\times S^1\times S^1$, $(r,\vartheta)\mapsto (r,\cos\vartheta,\sin\vartheta,1,0)$ be a \emph{poloidal cut} and $\toribbon:(0,1)\times (0,2\pi)\to [0,1]\times S^1\times S^1$, $(r,\varphi)\mapsto(r,1,0,\cos\varphi,\sin\varphi)$ be a \emph{toroidal annulus} (or ribbon). Representatives for the dual basis classes of the second relative cohomology group are the wedge products
\begin{align}
  \poloform &:= \radgen \wedge \polgen,&
  \toroform &:= \radgen\wedge\torgen,
\end{align}
with the property that $\int_{\toribbon} \toroform = 1 $, $\int_{\polocut}\poloform = 1$, $\int_{\toribbon}\sigma_{\varphi}=0$ and $\int_{\polocut}\sigma_{\vartheta}=0$. The third relative cohomology group is generated by forms cohomologous to the wedge products
\begin{align*}
  \poloform\wedge\torgen = -\toroform\wedge\polgen 
  = \radform\wedge \radgen
  = \radgen\wedge\polgen\wedge\torgen
\end{align*}
with the property that
\begin{align}
  \label{eq:vol-ht}
 \int\limits_{[0,1]\times S^1\times S^1}
\radgen\wedge\polgen\wedge\torgen=1. 
\end{align}

\section{Hodge decomposition of forms on toroidal volumes}
\label{sec:hodge-decomp}

From the purely homological observations of section \ref{sec:cohomology}, we draw a series of consequences on the Hodge decomposition of $k$-forms on the toroidal volumes $\ST$ and $\HT$. We refer here to appendix~\ref{sec:hodge-theory}.

\subsection{On a solid toroidal volume in $\R^3$}
{\nochange 


%

\begin{proposition}
  \label{prop:closed-2form}
  All closed $2$-forms on a solid toroidal volume $\ST$ are exact, i.e.  $B \in W^1\Omega^2(\ST)$, $dB=0 \iff B = d A$ where $A\in W^1\Omega^1(\ST)$.
\end{proposition}
This corresponds to the usual assertion that a divergence-free magnetic field on a solid toroidal volume can be expressed as the \emph{curl} of the vector potential.

\begin{proof} This follows immediately from the fact that the second cohomology group $H_{dR}^2(\ST,d)\cong H^2(\ST,\R) \cong \{0\}$ is trivial.
\end{proof}


\begin{corollary}
  All co-closed $1$-forms on $\ST$ are co-exact, i.e.  $\tilde{B}\in \Omega^1(\ST)$, $\delta\tilde{B} = 0 \iff \tilde{B} = \delta \tilde{A}$ where $\tilde{A}\in W^1\Omega^2(\ST)$.
\end{corollary}

\begin{proof}
By Hodge duality, Proposition \ref{prop:hodge-duality}, $H_{dR}^1(\ST,\delta)\cong H_{dR}^2(\ST,d)\cong \{0\}$.
\end{proof}

\begin{proposition}
  \label{prop:vacuum-field}
  On a solid toroidal volume, the space of Neumann harmonic $1$-fields $\Harm_N^1(\ST)$ is a one-dimensional vector space, spanned by a single field (unit vacuum toroidal field) $ \torfield$, defined up to sign by the following properties:
  \begin{enumerate}
  \item closed and co-closed, $d\torfield=0$ and $\delta\torfield=0$;
  \item everywhere tangential to the boundary, $\bm{n}\torfield=0$;
  \item unit $L^2$-norm, $||\torfield||_{L^2}=1$.
  \end{enumerate}
  The unit vacuum toroidal field is a purely geometric object and depends smoothly on the embedding $\Phi$.
\end{proposition}

\begin{proof}
  $\Harm_N^1(\ST)\cong H_{dR}^1(\ST,d)\cong H^1(\ST,\R) \cong \R$. The first two properties define any Neumann harmonic $1$-field and the third is a normalisation.
\end{proof}

Multiples of the unit vacuum toroidal field $\torfield$ correspond to the only closed, co-closed and tangential $1$-forms that are not exact; they cannot be expressed as the exterior derivative of smooth functions, since $\bra \torfield,df\ket = 0$, $\forall f\in  W^1C^\infty(M)$.  The normalisation of $\torfield$ is arbitrary; one could equally well choose $\torfield$ to be dual with respect to loop integrals (with the drawback of having to divide by the $L^2$-norm in the decomposition formulae below).

The first de Rham cohomology group of solid toroidal volumes can now be explicitly written as $H^1_{dR}(\ST,d)=\{[c\torfield]\ |\ c\in \R\} = \text{span}\{[\torfield]\}$, the set of equivalence classes represented by multiples of the unit vacuum toroidal field.

\begin{corollary}
  \label{prop:dirichlet-2forms}
  The space of Dirichlet harmonic $2$-fields $\Harm^2_D(\ST)$ is a one-dimensional vector space, spanned by $\star \torfield$ (the unit vacuum poloidal flux). As a closed $2$-form, the latter can be written as an exact form $\star\torfield = d\Upsilon$ with $\delta d\Upsilon=0$ and $\bm{t}d\Upsilon = 0$.
\end{corollary}

\begin{remark}
  The tangential operator commutes with the exterior derivative, so that $d\bm{t} \Upsilon = \bm{t}d \Upsilon = 0$. However, this does not mean that $\bm{t}\Upsilon=0$. In fact,
  \begin{align*}
    1&=||\star \torfield||^2_{L^2} = \int_{\ST} \torfield\wedge d\Upsilon
    =- \int_{\ST} d(\torfield\wedge \Upsilon)\\
    &= \int_{\partial \ST} \bm{t}\Upsilon\wedge \bm{t}\torfield = \int_{\partial \ST} \bm{t}\Upsilon\wedge\torfield
  \end{align*}
  where the last step follows from the fact that the unit vacuum toroidal field is purely tangential on the boundary. It is thus necessary that $\bm{t}\Upsilon\neq 0$.
\end{remark}

\begin{proposition}[Hodge decomposition of $1$-forms on a solid torus]
  \label{prop:closed-oneform}
  All $1$-forms $A \in L^2\Omega^1(\ST)$ can be orthogonally decomposed as
  \begin{align*}
    A = \delta \beta + df_D + dh + c \torfield,
  \end{align*}
  where $\beta\in W^1\Omega_N^2(\ST)$, $f_D\in  W^1C^\infty(\ST)$ is a smooth Dirichlet function on $\ST$ such that $f_D|_{\partial \ST}=0$, $h\in  W^1C^\infty(\ST)$ is a harmonic function such that $\Delta h = 0$, $\torfield\in \Harm^1_N(\ST)$ is the unit vacuum toroidal field and $c=\bra A,\torfield\ket$ is the projection of $A$ onto it.
\end{proposition}

\begin{proof}
  The decomposition of $1$-forms on $\ST$ immediately follows from HFMD and proposition \ref{prop:vacuum-field}. We note that for smooth functions $f\in \Omega^0(M)$, $\delta f = 0$, so that the Laplace-de Rham operator reduces to $\Delta f := (\delta d + d\delta) f = \delta d f$.
\end{proof}

\begin{theorem}[Minimal gauge on a solid torus]
  \label{thm:minimal_gauge}
  Let $B=dA\neq 0$ be a closed (and exact) $2$-form on $\ST$ as in proposition \ref{prop:closed-2form}. The potential $1$-form $A$ can be chosen uniquely to be minimal in terms of its $L^2$-norm by requiring that it be co-closed $\delta A=0$, tangential $\bm{n}A=0$ and whirl-free $\bra A,\torfield\ket = 0$.  In this gauge, $A$ carries the minimal amount of physical information about $B$.
\end{theorem}

Physical information refers here to the use of the $L^2$-norm to quantify the energy of fields, for example \emph{magnetic energy} by $||B||_{L^2}^2$.

\begin{proof}
  By proposition \ref{prop:closed-oneform}, $A=\delta \beta + df_D + dh + c\torfield$ so that $B=d\delta \beta$. A necessary and sufficient condition for $B\neq 0$ is $\delta\beta\neq 0$. By orthogonality of the decomposition, the $L^2$-norm squared of $A$ is the sum of positive quantities
  \begin{align*}
    ||A||^2_{L^2} = ||\delta \beta||_{L^2}^2 + || df_D||_{L^2}^2 + || dh||_{L^2}^2 + c^2 > 0
  \end{align*}
  If $A$ is co-closed, $\delta A=0$, proposition \ref{prop:closed-exact} shows that the component $df_D=0$. If $A$ is tangential, $0=\bm{n}A=\bm{n}dh$, but then
  \begin{align*}
    ||dh||_{L^2}=\bra dh,dh\ket = \bra h , \cancel{\Delta h} \ket + \int_{\partial \ST} \bm{t}h\wedge\star\cancel{\bm{n}dh} = 0
  \end{align*}
  and in fact $dh=0$. Finally, provided that $0 = \bra A,\torfield\ket = c$, $||A||^2_{L^2} = ||\delta\beta||_{L^2}^2$ is the minimal possible $L^2$-norm for which $B\neq0$.
\end{proof}

The requirement that the potential $1$-form $A$ be co-closed translates to the familiar requirement that the vector potential be divergence-free, which is commonly known as the Coulomb gauge. In contrast to the standard setting of $\R^3$, the two additional requirements are important because the solid toroidal volume is a manifold with boundary (hence the condition on the normal component) and its homology is non-trivial (hence the condition on the toroidal vacuum component).}

The Coulomb gauge (divergence-free vector potential) is rarely adopted in the context of toroidal magnetic confinement fusion and, to the authors' best knowledge, the tangential and whirl-free condition have never been applied. The minimal gauge, despite its mathematical attractiveness, has thus never been exploited. From a computational perspective, this implies that most vector potentials carry a significant fraction of physically irrelevant information. One would hope that computational effort is not wasted on those components.

{\nochange 
\begin{proposition}
  \label{prop:loop-integral}
  The integral of any closed $1$-form, $s\in W^1\Omega^1(\ST)$ such that $ds=0$, over any closed path is
  \begin{align*}
    \int_\Gamma s &=  (s,\torfield)m \torloop  , & m\in \Z,
  \end{align*}
  where $\torloop\neq 0 \in \R$ is a fixed coefficient that depends smoothly on the embedding $\Phi$ (normalisation of $\torfield$).
\end{proposition}

\begin{proof}
  The closed $1$-form $s$ is cohomologous to $c\torfield$ with $c=\bra s,\torfield\ket$. Thus,
  \begin{align*}
    \int_\Gamma s = c \int_\Gamma\torfield.
  \end{align*}
  The closed path $\Gamma$ is homologous to integer copies of $\Phi(\torpath)$. On $D^2\times S^1$, $\Phi^*(\torfield)$ is closed and is thus cohomologous to $\torloop \torgen$, with $\torloop\in \R$ computed by the loop integral
  \begin{align*}
    \torloop = \int_{\torpath}\Phi^*(\torfield)
    =\int_{\Phi(\torpath)}\torfield.
  \end{align*}
  Because $\torfield$ is unique (up to a sign), the factor $\torloop$ depends only on the embedding.
\end{proof}

\begin{proposition}
  \label{prop:flux-vac}
  The integral of a closed Dirichlet $2$-form, namely $B\in W^1\Omega_D^2(\ST)$ such that $dB=0$ and $\bm{t}B=0$, over a \emph{poloidal cut} $\Sigma_P\in [\Phi(\polocut)]\in H_2(\ST,\partial \ST)$, is equal to
  \begin{align*}
    \int_{\Sigma_P} B = \bra \star \torfield, B \ket \torflux
    = \torflux \int_{\ST}  \torfield \wedge B
  \end{align*}
  where $\torflux\neq 0\in \R$ is a fixed signed coefficient that depends smoothly on the embedding (through the normalisation of $\torfield$).
\end{proposition}

\begin{proof}
  The closed Dirichlet $2$-form belongs to the equivalence class $B\in [\tilde{c}\star \torfield]\in H_{dR}^2(\ST,\partial \ST,d)$ of the second relative de Rham cohomology, with $\tilde{c}=\bra B,\star\torfield\ket$. Thus,
  \begin{align*}
    \int_{\Sigma_P} B = \tilde{c}\int_{\Sigma_P}\star\torfield = \tilde{c}\int_{\Phi(\polocut)}\star \torfield
  \end{align*}
  since the poloidal cut $\Sigma_P$ is homologous to $\Phi(\polocut)$. On $D^2\times S^1$, $\Phi^*(\star \torfield)$ is closed and Dirichlet since
  \begin{align*}
    \bm{t}\Phi^*(\star \torfield)
    = \Phi^*(\bm{t}\star\torfield)
    =\Phi^*(\star\bm{n}\torfield)=\Phi^*(0)=0.
  \end{align*}
  Hence, $\Phi^*(\star \torfield)\in [\torflux \poloform]$, with $\torflux\in \R$ computed by the surface integral
  \begin{align*}
    \torflux = \int_{\polocut} \Phi^*(\star \torfield) = \int_{\Phi(\polocut)}\star\torfield .
  \end{align*}
  Because $\star\torfield$ is unique (up to a sign), the factor $\torflux$ depends only on the embedding.
\end{proof}

\begin{theorem}
  \label{thm:product-wedge}
  Let $s\in W^1\Omega^1(\ST)$ be a closed $1$-form, $ds=0$ and $B\in W^1\Omega_D^2(\ST)$ be a closed Dirichlet $2$-form, $dB=0$ and $\bm{t}B=0$. Then,
  \begin{align*}
\int_{\ST} s\wedge B =     \bra s,\star^{-1} B\ket= \int_{\Gamma_T} s \int_{\Sigma_P} B,
  \end{align*}
  where $\Gamma_T\in [\Phi(\torpath)]\in H_1(\ST)$ is a closed toroidal path and $\Sigma_P\in[\Phi(\polocut)]\in H_2(\ST,\partial \ST)$ is a poloidal cut.
\end{theorem}
}

\begin{proof} 
  The underlying reason is that the exterior product of closed differential forms induces the cup product in de Rham cohomology, so that $[\alpha]\smile[\beta] \equiv [\alpha\wedge\beta]$.

  Explicitly, the pullback $\Phi^* s\in [l\torgen]\in H_{dR}^1(D^2\times S^1,d)$ with
  \begin{align*}
    l=\int_{\Gamma_T} s=\int_{\torpath} \Phi^*s .
  \end{align*}
  The pullback $\Phi^*B\in [\tilde{l}\poloform]\in H_{dR}^2(D^2\times S^1,\partial D^2\times S^1,d)$ with
  \begin{align*}
   \tilde{l} = \int_{\Sigma_P} B= \int_{\polocut}\Phi^*B.
  \end{align*}
  The pullback of their wedge product belongs to the same class as $[\Phi^*(s\wedge B)] = [\Phi^*s\wedge\Phi^*B]= l\tilde{l}[\torgen\wedge \poloform] \in H_{dR}^3(D^2\times S^1,\partial D^2\times S^1,d)$. Hence,
  \begin{align*}
    \int\limits_{ST} s\wedge B
    &= \int\limits_{\Phi(D^2\times S^1)} s\wedge B
    = \int\limits_{D^2\times S^1} \Phi^*(s)\wedge\Phi^*(B)\\
    &=l\tilde{l}\int\limits_{D^2\times S^1} \torgen\wedge\poloform = l\tilde{l}
  \end{align*}
where the last step follows from equation (\ref{eq:wedge-canonical}).
One subtle point in the first step, where $\ST$ is replaced by $\Phi(D^2\times S^1)$ for the purpose of integration over the whole volume, is that the embedding must preserve the orientation of the manifold with respect to the standard volume form $dX\wedge dY\wedge dZ\in \Omega^3(\ST\subset \R^3)$, otherwise an extra minus sign must be included in the formula.
\end{proof}

The beauty here (and power of the homological approach) is that the result neither depends on the explicit representative forms $\torgen$ and $\poloform$, nor on the specific choice of embedding $\Phi$ (up to a minus sign if the manifold's orientation is reversed). It means that the factorisation $\int s\wedge B = \int s \int B$ occurs no matter the local coordinates used to represent the fields, which is particularly useful for numerical applications, either as a method of verifying the consistency and accuracy of the solver, or as a method of by-passing costly volume integrals. We warn that an exact match between $\int s\wedge B = \int s \int B$ is sensitive to round-off error, such that the order of numerical operations may matter in the intricate balance of volume integrals via quadrature rules.

{\nochange
\begin{corollary} The product of coefficients $\torloop$ and $\torflux$ is unity
  \begin{align}
    \torloop\torflux = 1.
  \end{align}
\end{corollary}

\begin{proof} Applying the result of Theorem \ref{thm:product-wedge} to the unit vacuum toroidal field, we have
  \begin{align*}
    1= \int_{\ST}\torfield\wedge\star\torfield
    =\int_{\Gamma_T}\torfield\int_{\Sigma_P}\star\torfield
    = \torloop\torflux.
  \end{align*}

\end{proof}

}

\subsection{The illustrative example of an axisymmetric embedding}
To illustrate the results of the previous section, we consider the special case of an \emph{axisymmetric} solid toroidal volume, such as the standard tokamak configuration. The adjective axisymmetric means that the embedding maps the disk independently from the circle and takes the form $\Phi_{axi}:D^2\times S^1\to \R^3$, $(x_1,x_2,x_3,x_4)\mapsto (X,Y,Z)=(R(x_1,x_2) x_3,R(x_1,x_2)x_4,Z(x_1,x_2))$ where $(R,Z):D^2\to \Sigma_P\subseteq [R_{min},R_{max}]\times [Z_{min},Z_{max}]$ are functions referencing the major radius and the vertical position. The major radius matches the function $R(X,Y,Z)=\sqrt{X^2+Y^2} = R(x_1,x_2)$, with the requirement that $R_{min}>0$.

Locally, one can define the \emph{toroidal angle} $\varphi(X,Y,Z)$, such that $X= R \cos\varphi $ and $Y=R\sin\varphi$. Its value coincides with the local angle $\varphi$ on $S^1$, such that $x_3=\cos\varphi$ and $x_4=\sin\varphi$. The toroidal angle is not a smooth function on $\ST_{axi}\subset \R^3$, yet it is common to denote $\gradphi = (-YdX+XdY)/(X^2+Y^2)$, a (global) $1$-form on $\ST_{axi}$, as its local exterior derivative.
%
%
The pullback is a multiple of the representative defined in equation (\ref{eq:rep_phi}) of the first cohomology generating class, $\Phi_{axi}^*(\gradphi) = 2\pi \torgen$. Hence, $\gradphi$ is closed but deceivingly not exact since $\int_{\Phi_{axi}(\torpath)} \gradphi=\int_{\torpath} \Phi_{axi}^*(\gradphi) = 2\pi \int_{\torpath} \torgen = 2\pi$. This means that $[\gradphi]\neq[0]\in H^1_{dR}(\ST_{axi},d)$. 

In addition, it can be verified that $\gradphi$ is co-closed with respect to the Euclidean metric, $\delta \gradphi=0$, and tangential to the boundary of $\ST_{axi}$, $\bm{n}\gradphi=0$ (specific to axisymmetric embeddings $\Phi_{axi}$). These properties actually make the $1$-form $\gradphi$ a Neumann harmonic field, $"d\varphi "\in \Harm^1_N(\ST_{axi})$. The unit vacuum toroidal field on $\ST_{axi}$ is then simply $\torfield = \gradphi/||\gradphi||_{L^2}$ where
\begin{align*} 
  ||\gradphi||^2_{L^2}&= \int_{\ST_{axi}} \langle \gradphi,\gradphi\rangle \mu
  = 2\pi \int_{\Sigma_P} R^{-1}dRdZ,
\end{align*}
with $\Sigma_P$ being a \emph{poloidal cut} of the torus. If the latter domain is simply a rectangle $\Sigma_P = [R_{min},R_{max}]\times[Z_{min},Z_{max}]$, the $L^2$-norm reduces to $||\gradphi||^2_{L^2} = 2\pi (Z_{max}-Z_{min}) \ln(R_{max}/R_{min})$. 

When expressing the vector potential $1$-form in so-called \emph{covariant components} with respect to cylindrical coordinates, $A = A_RdR + A_ZdZ + A_\varphi \gradphi$, the whirl-free condition from the minimal gauge Theorem~\ref{thm:minimal_gauge} reads
\begin{align*}
  0& = (A,\gradphi)  \iff 
\int_{\Sigma_P} \frac{A_\varphi(R,Z)}{R} dR dZ = 0.
\end{align*}

For general embeddings $\Phi$, computing the unit toroidal vacuum field is not trivial (although it can be cast as the solution to a variational problem). The following proposition provides a recipe when $\Phi$ is not "substantially different" from axisymmetric.
\begin{proposition}
If $\ST$ is isotopic to $\ST_{axi}$ and there exists an isometry $\mathcal{I}$ such that $\ST\cap \mathcal{I}(\{(0,0,Z)\in \R^3\}) =\varnothing $, the unit toroidal vacuum field can be expressed as $\torfield =(\gradphi - dh)/||\gradphi-dh||_{L^2}$ where $\gradphi=\mathcal{I}^*[(-YdX+XdY)/(X^2+Y^2)]$ and $h\in W^1 C^\infty(\ST)$ is the unique solution (up to a constant) to the boundary value problem
\begin{align*}
\Delta h &= 0 ,&
\bm{n}dh &= -\bm{n}\gradphi.
\end{align*}
\end{proposition}

\begin{proof}
With $\ST\cap \mathcal{I}(\{(0,0,Z)\in \R^3\}) =\varnothing $, $\gradphi$ is a well-defined $1$-form on $\ST$. It is closed and co-closed with respect to the Euclidean metric, because $\mathcal{I}$ preserves the Hodge star. Hence, $\gradphi\in \Harm^1(\ST)$ and, by virtue of HD, can be decomposed as $\gradphi = dh + c \torfield$ where $h\in W^1C^\infty(\ST)$ is a harmonic function $\Delta h=0$, $\torfield\in \Harm_N^1(\ST)$ is the unit vacuum toroidal field and $c = \bra \gradphi,\nu_{\varphi}\ket$. Since $\bm{n}\torfield =0$, the harmonic function $h$ must compensate the normal component of $\gradphi$ on the boundary, $\bm{n}dh=-\bm{n}\gradphi$. This provides sufficient boundary conditions to uniquely determine $h$ (up to a constant). The isotopy ensures through the invariance of $2\pi = \int_{\Phi(\torpath)} \gradphi = c \torloop$ that $c\neq 0 $ and thus $||\gradphi-dh||_{L^2}>0$.
\end{proof}

\subsection{Hodge decomposition on a hollow toroidal volume in $\R^3$}
{\nochange
Arguing similarly as for the Hodge decomposition of $k$-forms on solid toroidal volumes, we obtain the following on a hollow torus.
\begin{proposition}
  \label{prop:closed-2form_ht}
  A closed $2$-form $B\in W^1\Omega^2(\HT)$, $dB=0$, on a hollow torus can be expressed as
  \begin{align*}
    B =  dA + c \radfield
  \end{align*}
  where $A\in W^1\Omega^1(\HT)$, and $\radfield \in \Harm_N^2(\HT)$ (the unit vacuum radial flux) is such that $d\radfield=0$, $\delta \radfield=0$, $\bm{n}\radfield=0$ and $|| \radfield||_{L^2}=1$. The coefficient $c=\bra B,\radfield\ket $.
\end{proposition}

\begin{proof}
  The second de Rham cohomology group of $\HT$ is one-dimensional $\R \cong H^2(\HT,\R)\cong H_{dR}^2(\HT,d)=\text{span}\{[\radfield]\}$. 
\end{proof}


\begin{proposition}
  \label{prop:radintegral}
  The integral of a closed $2$-form, $B\in W^1\Omega^2(\HT)$ with $dB=0$, over a toroidal surface $\radform\in [\Phi(T^2)]\in H_2(\HT)$ is equal to
  \begin{equation*}
    \int_{\radform} B =  \bra B, \radfield \ket \radflux,
  \end{equation*}
  where $\radflux\neq 0\in \R$ is a fixed coefficient that depends smoothly on the embedding (through the normalisation of $\radfield$).
\end{proposition}

\begin{proof}
  The closed $2$-form $B$ is cohomologous to $c\radfield$ where $c=\bra B,\radfield\ket$.
  \begin{align*}
    \int_{\radform} B = c \int_{\Phi(T^2)} \radfield.
  \end{align*}
  The pullback of the unit vacuum radial flux is closed and is thus cohomologous to $\radflux \radform$, with $\radflux \in \R$ computed by
  \begin{align*}
   \radflux = \int_{T^2} \Phi^*(\radfield) = \int_{\Phi(T^2)} \radfield.
 \end{align*}
 Because $\radfield$ is unique (up to a sign), the factor $\radflux$ depends only on the embedding.
\end{proof}

\begin{corollary}
  \label{prop:ht-closedexact}
  Closed Dirichlet $2$-forms are exact (in the absolute de Rham cohomology), i.e $B\in W^1\Omega_D^2(\HT)$, $dB=0$ and $\bm{t}B=0 \iff B=dA$ for some $A\in W^1\Omega^1(\HT)$.
\end{corollary}
\begin{proof}
  The surface $\radform\sim \Phi(T^2_0)$ is homologous to the interior (or exterior) boundary on which the flux evaluates as the zero form. Thus, $\int_{\radform} B = \int_{\Phi(T^2_0)} \cancel{\bm{t} B} = 0$ and $\bra B, \rho\ket = 0 $ by proposition \ref{prop:radintegral}. The claim follows then from proposition \ref{prop:closed-2form_ht}.
\end{proof}
}

Proposition \ref{prop:closed-2form_ht} shows that a divergence-free vector field on a hollow toroidal volume is not necessarily the curl of a vector potential. Yet, from a physical perspective, the unit vacuum radial flux $\radfield$ cannot describe a real magnetic field. A contradiction with Gauss's law for magnetism over the whole of $\R^3$ arises by the fact that, as per Proposition~\ref{prop:radintegral}, its flux $\int_{\Phi(T_0^2)} \radfield = \Psi_R \neq 0$ through the interior boundary of the hollow torus does not vanish. But, this surface bounds a solid toroidal volume (the part that is removed to create the hollow torus), and if a physical magnetic field were to exist in there, its flux through the boundary would vanish. In other words, the unit vacuum radial flux, although mathematically admissible as solution to the boundary value problem, is unrealisable as a magnetic field. From this physical consideration, it is concluded that the magnetic field on a hollow toroidal volume derives from a vector potential, and that the magnetic flux belongs to the subspace of exact $2$-forms on $\HT$.

{\nochange 
\begin{proposition}
  \label{prop:vacuum-field-ht}
  The space of Neumann harmonic $1$-fields $\Harm_N^1(\HT)$ is a two-dimensional vector space, spanned by the pair of basis vectors $\{\torfield,\polfield\}$ called the \emph{unit vacuum poloidal and toroidal fields}, defined uniquely (up to signs) by the following properties. Letting indices $i,j\in \{\varphi,\vartheta\}$, the basis vectors are
  \begin{enumerate}
  \item closed and co-closed, $d\nu_i=0$ and $\delta\nu_i=0$;
  \item everywhere tangential to the boundary, $\bm{n}\nu_i=0$;
  \item $L^2$-orthonormal, $\bra \nu_i,\nu_j\ket = \delta_{ij}$;
  \item $\int_{\Gamma_T} \polfield = 0$ and $\int_{\Gamma_P} \torfield = 0 $, where $\Gamma _T\in[\Phi(\torpath)]$ and $\Gamma _P\in [\Phi(\polpath)]$ are closed toroidal and poloidal paths with $H_1(\HT)=\text{span}\{[\Phi(\torpath)],[\Phi(\polpath)]\}$.
  \end{enumerate}
  The unit vacuum poloidal and toroidal fields are purely geometric objects and depend smoothly on the embedding $\Phi$.
\end{proposition}

\begin{proof}
  $\Harm_N^1(\HT)\cong H^1(\HT,\R) \cong \R^2$. The first two properties define any Neumann harmonic $1$-field. The last condition states that $\Phi^*\polfield = \polloop \polgen$ and $\Phi^*\torfield =\torloop \torgen$, where $\polloop = \int_{\Gamma_P}\polfield$ and $\torloop = \int_{\Gamma_T}\torfield$. The Hodge star of Neumann harmonic $1$-fields are Dirichlet $2$-fields. We list the pullbacks $\Phi^*\star \polfield = \polflux \toroform + \kappa_\vartheta \poloform$ and $\Phi^*\star\torfield = \kappa_\varphi \toroform + \torflux \poloform$. The orthogonality condition ensures that
  \begin{align*}
    0 &= \int\limits_{\HT} \polfield\wedge\star \torfield
    =\int\limits_{[0,1]\times S^1\times S^1}
    \polloop \left(\kappa_\varphi \polgen \wedge\toroform
      + \torflux \cancel{\polgen \wedge \poloform} \right)\\
    &= -\polloop \kappa_\varphi \iff \kappa_\varphi = 0 \\
0 &= \int\limits_{\HT} \torfield\wedge\star \polfield
    = \int\limits_{[0,1]\times S^1\times S^1}
    \torloop\left(\polflux \cancel{\torgen \wedge \toroform}
      + \kappa_\vartheta \torgen \wedge\poloform \right)\\
    &= \torloop \kappa_\vartheta \iff \kappa_\vartheta= 0,
  \end{align*}
  which implies that $\int_{\Sigma_T} \star \torfield = \kappa_\varphi = 0$ and $\int_{\Sigma_P} \star \torfield = \kappa_\vartheta = 0$ where $\Sigma_P\in[\Phi(\polocut)]$ is a poloidal cut and $\Sigma_T\in[\Phi(\toribbon)]$ is a toroidal annulus (or ribbon) with $H_2(\HT,\partial \HT)=\text{span}\{[\Phi(\toribbon)],[\Phi(\polocut)]\}$. Finally,
  \begin{align*}
    1 &= \int\limits_{\HT} \polfield\wedge\star \polfield
    = \int\limits_{[0,1]\times S^1\times S^1}\polloop\polflux \ \polgen\wedge\toroform
    = - \polloop\polflux\\
    1 &= \int\limits_{\HT} \torfield\wedge\star \torfield
    = \int\limits_{[0,1]\times S^1\times S^1}
    \torloop\torflux\ \torgen \wedge\poloform
    = \torloop \torflux,
  \end{align*}
  where $\polflux = \int_{\Sigma_T} \star \polfield$ and $\torflux = \int_{\Sigma_P}\star \torfield$, sets the normalisation with respect to the manifold's metric.
\end{proof}

\begin{proposition}
  \label{prop:hd-oneform-ht}
  All $1$-forms $A \in L^2\Omega^1(\HT)$ can be orthogonally decomposed as
  \begin{align*}
    A = \delta \beta + df_D + dh + c_\varphi \torfield + c_\vartheta\polfield,
  \end{align*}
  where $\beta\in W^1\Omega_N^2(\HT)$, $f_D\in  W^1C^\infty(\HT)$ is a smooth Dirichlet function on $\HT$ such that $f_D|_{\partial \HT}=0$, $h\in  W^1C^\infty(\HT)$ is a harmonic function such that $\Delta h = 0$, and $\torfield$ and $\polfield$ are the unit vacuum toroidal and poloidal fields respectively which span the space of Neumann harmonic fields $\Harm^1_N(\HT)$. The projection coefficients are $c_\varphi=\bra A,\torfield\ket$ and $c_\vartheta=\bra A,\polfield\ket$.
\end{proposition}

\begin{proof}
  The decomposition of $1$-forms on $\HT$ immediately follows from HFMD and proposition \ref{prop:vacuum-field-ht}.
\end{proof}

\begin{theorem}[Minimal gauge on a hollow torus]
  \label{thm:minimal_gauge-ht}
  Let $B=dA\neq 0$ be an exact $2$-form on $\HT$. The potential $1$-form $A$ can be chosen uniquely to be minimal in terms of its $L^2$-norm by requiring that it be co-closed $\delta A=0$, tangential $\bm{n}A=0$ and whirl-free, i.e. $\bra A,\torfield\ket = 0$ and $\bra A,\polfield \ket=0$.  In this gauge, $A$ carries the minimal amount of physical information about $B$.
\end{theorem}

\begin{proof}
  The proof is essentially the same as for theorem \ref{thm:minimal_gauge-ht}. By proposition \ref{prop:hd-oneform-ht}, $A=\delta \beta + df_D + dh + c_\vartheta \polfield + c_\varphi\torfield$. By orthogonality of the decomposition, the $L^2$-norm squared of $A$ is the sum of positive quantities
  \begin{align*}
    ||A||^2_{L^2} = ||\delta \beta||_{L^2}^2 + || df_D||_{L^2}^2 + || dh||_{L^2}^2 + c_\vartheta^2 + c_\varphi^2 > 0
  \end{align*}
  If $\delta A=0$, then by proposition \ref{prop:closed-exact} $df_D=0$. If $A$ is tangential, $0=\bm{n}A=\bm{n}dh$, but then $||dh||_{L^2}=0$ and so $dh=0$. Finally, provided that $0 = \bra A,\torfield\ket = c_\varphi$ and $0=\bra A,\polfield\ket = c_\vartheta$, $||A||^2_{L^2} = ||\delta\beta||_{L^2}^2$ is the minimal possible $L^2$-norm for which $B\neq0$.
\end{proof}

\begin{theorem}
  \label{thm:product-wedge-ht}
  Let $s\in W^1\Omega^1(\HT)$ be a closed $1$-form, $ds=0$ and $B\in W^1\Omega_D^2(\HT)$ be a closed Dirichlet $2$-form, $dB=0$ and $\bm{t}B=0$. Then,
  \begin{align*}
\int_{\HT} s\wedge B =     \bra s,\star^{-1} B\ket
    = \int_{\Gamma_T} s \int_{\Sigma_P}B - \int_{\Gamma_P}s \int_{\Sigma_T}B
  \end{align*}
  where $\Gamma_T\in [\Phi(\torpath)]$ and $\Gamma_P\in [\Phi(\polpath)]$ are closed toroidal and poloidal paths, $\Sigma_P\in[\Phi(\polocut)]$ is a poloidal cut and $\Sigma_T\in[\Phi(\toribbon)]$ is a toroidal annulus (or ribbon).
\end{theorem}
}
\begin{proof}
  The pullback $\Phi^*s\in [l_\vartheta\polgen + l_\varphi \torgen]\in H^1_{dR}([0,1]\times S^1\times S^1,d)$ with
  \begin{align*}
    l_\vartheta &= \int_{\Gamma_P} s = \int_{\polpath}\Phi^* s, &
    l_\varphi &= \int_{\Gamma_T} s =\int_{\torpath}\Phi^*s.
  \end{align*}
  The pullback $\Phi^*B\in [\tilde{l}_\varphi\toroform + \tilde{l}_\vartheta \poloform]\in H^2_{dR}([0,1]\times S^1\times S^1,\{0\}\times S^1\times S^1\cup \{1\}\times S^1\times S^1,d)$ with
  \begin{align*}
    \tilde{l}_\vartheta &= \int_{\Sigma_P} B=\int_{\polocut} \Phi^*B , &
   \tilde{l}_\varphi &= \int_{\Sigma_T} B =
  \int_{\toribbon}\Phi^*B 
  \end{align*}
  Then, recalling equation (\ref{eq:vol-ht}),
  \begin{align*}
    \int_{HT}s\wedge B&
    = \int\limits_{[0,1]\times S^1\times S^1} \Phi^*(s)\wedge \Phi^*(B) \\
    &= \int\limits_{[0,1]\times S^1\times S^1} \left(
      l_\vartheta\tilde{l}_\varphi \polgen\wedge\toroform
      + l_\varphi\tilde{l}_\vartheta\torgen\wedge\poloform
    \right)\\
    &= (l_\varphi\tilde{l}_\vartheta- l_\vartheta \tilde{l}_\varphi)\int\limits_{[0,1]\times S^1\times S^1} \radgen\wedge\polgen\wedge\torgen\\
    &= l_\varphi\tilde{l}_\vartheta- l_\vartheta \tilde{l}_\varphi
    = 
     \begin{pmatrix}
    l_\varphi & l_\vartheta
    \end{pmatrix}
    J
    \begin{pmatrix}
    \tilde{l}_\varphi \\ \tilde{l}_\vartheta
    \end{pmatrix},
  \end{align*}
  where $J=
    \begin{pmatrix}
    0 & 1 \\
    -1 & 0
    \end{pmatrix}$. The first step, where $\HT$ is replaced by $\Phi([0,1]\times S^1\times S^1)$ for the purpose of integration over the whole volume, is justified by the assumption that $\Phi$ preserves orientation with respect to the standard volume form on $\R^3$.
\end{proof}

The independence of this result on the embedding is highlighted by the following construction. Let $\Phi_1$ and $\Phi_2$ be two embeddings of the hollow torus describing the same manifold $\HT$, namely the images $\Phi^1([0,1]\times S^1\times S^1)=\Phi^2([0,1]\times S^1\times S^1)=\HT$ describe the same hollowed toroidal volume in $\R^3$. The diffeomorphism $\Phi_1\circ \Phi_2^{-1}:\HT\to\HT$ induces the mapping $A:H_1(\HT,\Z)\cong \Z^2\to H_1(\HT)\cong \Z^2$ on the first homology with integer coefficients such that toroidal and poloidal paths (generators) may be related by 
\begin{align*}
    \begin{pmatrix}
    [\Phi^1(\torpath)]\\
    [\Phi^1(\polpath)]
    \end{pmatrix}
    &=\begin{pmatrix}
    A_{11} & A_{12} \\
    A_{21} & A_{22}
    \end{pmatrix}
    \begin{pmatrix}
[\Phi^2(\torpath)]\\ 
[\Phi^2(\polpath) ]
    \end{pmatrix},& 
     A_{ij}\in \Z
    \end{align*}
The inverse of the matrix $A\in GL(2,\Z)$ is induced by the diffeomorphism $\Phi_2\circ \Phi_1^{-1}$ so that $A^{-1}\in GL(2,\Z)$ must also have integer coefficients. This is possible if and only if $\det(A)=\det(A^{-1})=\pm 1$. Indeed, $1=\det(A A^{-1}) = \det(A)\det(A^{-1})$ where $\det(A)\in \Z$ and $\det(A^{-1})\in \Z$. We note that in $2$-dimension $A^T JA = \det(A) J$, so that $\det(A)=\pm 1$ is equivalent to $A$ being (pseudo)-symplectic. A negative sign indicates here a reversal of the "handedness" or \emph{intersection number} of the bases.

Let $[\Lambda^{1,2}_{\varphi,\vartheta}]\in H_{dR}^1(\HT,d)$ be dual bases dual to the first homology generators  $[\Phi^{1,2}(C_{T,P})]\in H_1(\HT)$ such that $[{\Phi^{1,2}}^*(\Lambda^{1,2}_{\varphi,\vartheta})]=[\lambda_{\varphi,\vartheta}]$. The two dual bases are then related by
\begin{align*}
    \begin{pmatrix}
    [\Lambda^1_\varphi]\\
    [\Lambda^1_\vartheta ]
    \end{pmatrix}
    &=A^T
    \begin{pmatrix}
[\Lambda^2_\varphi]\\ 
[\Lambda^2_\vartheta ]
    \end{pmatrix}.
    \end{align*}
    
On similar grounds, the first relative homology bases must be related by $[\Phi^1(\radpath)]=B [\Phi^2(\radpath)]$ where $B=\pm 1$. A negative sign corresponds to exchanging the inner and outer boundaries. If $[R^{1,2}]\in H^1_{dR}(\HT,\partial\HT,d)$ are dual bases of first relative cohomology, then $[R^1]=B [R^2]$.

Let $[\Pi^{1,2}_{\varphi,\vartheta}]$ be dual bases of second relative cohomology, such that $[{\Phi^{1,2}}^*(\Pi^{1,2}_{\varphi,\vartheta})] =[ \sigma_{\varphi,\vartheta}] = [\lambda_r\wedge \lambda_{\varphi,\vartheta}]$. The latter relation implies that $[\Pi^{1,2}_{\varphi,\vartheta}] = [R^{1,2}\wedge \Lambda^{1,2}_{\varphi,\vartheta}]$. Hence, the second relative cohomology bases are related through matrix $A^T$ in the same way as for the first absolute cohomology bases, up to the sign conveyed by $B$. We have
\begin{align*}
    \begin{pmatrix}
    [\Pi^1_\varphi]\\
    [\Pi^1_\vartheta ]
    \end{pmatrix}
    &=B A^T
    \begin{pmatrix}
[\Pi^2_\varphi]\\ 
[\Pi^2_\vartheta ]
    \end{pmatrix}.
\end{align*}
By duality, generators of second relative homology $[\Phi^{1,2}(S_{T,P})]\in H_2(\HT,\partial \HT)$ are related by 
\begin{align*}
    \begin{pmatrix}
    [\Phi^1(\toribbon)]\\
    [\Phi^1(\polocut)]
    \end{pmatrix}
    &=BA
    \begin{pmatrix}
[\Phi^2(\toribbon)]\\ 
[\Phi^2(\polocut) ]
\end{pmatrix}.
\end{align*}
Combined together, $[\Lambda_\varphi^{1,2}\wedge \Pi^{1,2}_\vartheta] = -[\Lambda_\vartheta^{1,2}\wedge \Pi^{1,2}_\varphi]\in H^3(\HT,\partial\HT,d)$ are top forms on $\HT$, related to one another by the factor $\det(A)B=\pm 1$. The orientation of both top forms must be compatible with the standard volume form in $\R^3$ as per our global assumption on the embeddings. This implies that the mutual orientation of the embeddings' top forms must be preserved, and so $\det(A)B=1$. Interestingly, it is still possible for the intersection number of the first homology bases to differ, together with an exchange of inner and outer boundaries. It would be straightforward to make the above formula completely general by including information about the manifold's orientation. 

Given a closed $1$-form $s\in H_{dR}^1(\HT)$ and a closed Dirichlet $2$-form $B\in H_{dR}^2(\HT,\partial \HT)$ as above, let $l^{1,2}_{\varphi,\vartheta} = \int_{\Phi^{1,2}(C_{T,P})}s$ and $\tilde{l}^{1,2}_{\varphi,\vartheta} = \int_{\Phi^{1,2}(S_{T,P})}B$ be the projections on the first and second cohomology bases. It is then clear that these coefficients are related in the same way that the homology bases are, namely
\begin{align*}
    \begin{pmatrix}
    l^1_\varphi\\
    l^1_\vartheta 
    \end{pmatrix}
    &=A
    \begin{pmatrix}
l^2_\varphi\\ 
l^2_\vartheta
    \end{pmatrix},&
    \begin{pmatrix}
    \tilde{l}^1_\varphi\\
    \tilde{l}^1_\vartheta 
    \end{pmatrix}
    &=BA
    \begin{pmatrix}
\tilde{l}^2_\varphi\\ 
\tilde{l}^2_\vartheta
    \end{pmatrix}.
\end{align*}

We now verify that the factorisation result of Theorem \ref{thm:product-wedge-ht} is independent of the embedding by computing
\begin{align*}
    \int\limits_{\Phi^1([0,1]\times S^1\times S^1)} s\wedge B 
& =   \begin{pmatrix}
    l^1_\varphi & l^1_\vartheta
    \end{pmatrix}
    J
    \begin{pmatrix}
    \tilde{l}^1_\varphi \\ \tilde{l}^1_\vartheta
    \end{pmatrix}\\
    &= B \begin{pmatrix}
    l^2_\varphi & l^2_\vartheta
    \end{pmatrix} 
    A^T J A
    \begin{pmatrix}
    \tilde{l}^2_\varphi & \tilde{l}^2_\vartheta
    \end{pmatrix}\\
    &= \det(A)B\begin{pmatrix}
    l^2_\varphi & l^2_\vartheta
    \end{pmatrix}
    J
     \begin{pmatrix}
    \tilde{l}^2_\varphi & \tilde{l}^2_\vartheta
    \end{pmatrix}
    \\
    &=\int\limits_{\Phi^2([0,1]\times S^1\times S^1)} s\wedge B,
\end{align*}
This result is reminiscent of \emph{Riemann's bilinear relations} for compact Riemann surfaces of genus $g$, where $g=1$ in our case.

\section{Magnetic helicity on toroidal volumes}
\label{sec:helicity}
{\nochange
Assuming the magnetic flux $B$ to be an exact $2$-form $B=dA$, we define the magnetic helicity density as the top form on $M$
\begin{equation}
  \label{eq:helicity}
 H(A):=A\wedge dA.
\end{equation}
It can either be seen as a property derived from the potential $1$-form or, a fundamental gauge-dependent quantity related to the magnetic field. On a solid toroidal volume, this expression is well-defined for all closed $2$-forms, as per proposition \ref{prop:closed-2form}. On a hollow torus, not all closed $2$-forms are exact so that the notion of magnetic helicity only makes sense on the subspace of exact $2$-forms.

The definition of magnetic helicity density does not depend on the Riemannian metric, but can be written as the dot-product between the magnetic field and the vector potential,  $H(A)= \bm{A}\cdot\bm{B}\,\mu$ where $\mu$ is the natural volume form on $M$. The integral of the helicity density over the entire manifold is called the total helicity
\begin{align}
  K[A] :&= \int_{M} H(A) = \int_{M} A\wedge dA = \int_{M} \bm{A}\cdot\bm{B}\, \mu
\end{align}

The total helicity is important because it is one of several conserved quantities under ideal MHD motion called Casimirs~\cite{holm-marsden-ratiu,ono-1995,hattori-1994}. The total helicity is interpreted as the degree of knottedness of magnetic field-lines\cite{moreau-1961,moffatt-1969,arnold-1974}, which can be seen as a topological invariant. In a system with finite resistivity, energy dissipation is more rapid than the decay of total helicity~\cite{taylor-1974}. Taylor-relaxed states are critical solutions (within the same gauge) of the energy functional $E[A]=\frac{1}{2}\bra dA,dA\ket=\frac{1}{2}\int_{M} B^2 \mu$ subject to the constraint of fixed total helicity.

It is clear that the helicity density is a gauge-dependent quantity. If $A'$ and $A$ yield the same magnetic flux $B=dA'=dA$, their difference is a closed $1$-form $s:=A'-A$, $ds=0$. Then, the change in helicity density is $H(A')-H(A)=s\wedge B$ and the total helicity differs by the integral
\begin{align}
  K[A']-K[A]&= \int_{M} s\wedge B.
\end{align}
On a manifold without boundary, the total helicity is gauge-invariant by Stokes' theorem since $s\wedge B = d(A\wedge s)=d(A'\wedge s)$ is exact. On a manifold with boundary, the total helicity is a boundary term that can generally be altered by the gauge.

On a solid toroidal volume, HFMD of a closed $1$-form $s$ gives $s = df_D + dh_s + c_\varphi\torfield$ where $c_\varphi=(s,\torfield)$. Since $df\wedge B = d(fB) - f\cancel{dB}$ for any smooth function $f$, the change in total helicity consists of two terms
\begin{align}
  \label{eq:st-helicity}
  K[A']-K[A]
  &= \int_{\partial \ST} h_s\bm{t}B + \bra s,\torfield\ket \bra \star \torfield,B\ket
\end{align}
where the surface integral involving $f_D$ vanishes due to the Dirichlet boundary condition.

On a hollow toroidal volume, $s = df_D + dh_s + c_\vartheta\polfield + c_\varphi\torfield$ with $c_\vartheta=(s,\polfield)$ and $c_\varphi=(s,\torfield)$. By a similar argument, the difference in total helicity is decomposed as
\begin{multline}
  \label{eq:ht-helicity}
  K[A']-K[A]
  = \int_{\partial \HT} h_s\bm{t}B  + \bra s,\polfield\ket\bra\star\polfield,B\ket\\
   + \bra s,\torfield\ket \bra \star \torfield,B\ket.
\end{multline}

\subsection{Perfectly conducting boundary on a solid toroidal volume}
In the special case where the magnetic field is everywhere tangential to the boundary of the solid toroidal volume, $\bm{B}\cdot \bm{n}=0$, the magnetic flux $B$ is a closed Dirichlet $2$-form, $dB=0$ and $\bm{t}B=0$. Then, the first term in equation (\ref{eq:st-helicity}) vanishes and Theorem \ref{thm:product-wedge} greatly simplifies the computation of the second. Such boundary condition on the magnetic field is known as the \emph{perfectly or ideally conducting} boundary. It is often assumed in tokamak and stellarator physics, and describes flux tubes as regions bounded by magnetic surfaces, for example in solar flare dynamics. By virtue of Theorem \ref{thm:product-wedge}, the total magnetic helicity can be evaluated as the product of the integral of the gauge function over any closed toroidal path $\Gamma_T\in[\Phi(\torpath)]\in H_1(\ST)$ times the flux of the magnetic field through any poloidal cross-section $\Sigma_P\in[\Phi(\polocut)]\in H_2(\ST,\partial\ST)$,
\begin{align}
  K[A']-K[A] &= \int_{\Gamma_T} (A'-A) \int_{\Sigma_P} B.
\end{align}

\begin{remark}
  It is sufficient for two vector potentials to agree along a single closed toroidal path to yield the same value of total magnetic helicity. This condition is for example met in the context of toroidal magnetic confinement, such as tokamaks and stellarators, when the vector potential is considered to vanish on the so-called \emph{magnetic axis}. This convenient choice is adopted by several magnetostatics codes~\cite{vmec,spec}, for which the values of total magnetic helicity can be compared directly (provided that the magnetic fields are tangential to the same boundary, and the magnetic axes coincide in real space).
\end{remark}

By rearranging the terms in the above equation, one can define a \emph{gauge-invariant} or \emph{relative} total helicity,
\begin{align}
  K_r[A;\Gamma_T] :&= K[A] - \int_{\Gamma_T} A \int_{\Sigma_P} B .
\end{align}
The relative total helicity does not depend on the choice of gauge, $K_r[A';\Gamma_T]-K_r[A;\Gamma_T]=0$, but may depend on the closed toroidal path $\Gamma_T$. The latter must be the same when comparing two such quantities, which means that the reference poloidal flux must coincide.
Restricting to toroidal paths $\Gamma_T$ that lie on the boundary of the solid torus, the loop integral of $A$ is independent of the specific choice of path. Indeed, by Stokes' theorem, the difference equates to a surface integral of $\bm{t}B$ which vanishes due to the perfectly conducting boundary. In this case, the relative helicity reads
\begin{align}
 K^{ext}_r[A]: = K[A] - \Psi_P^{ext}\Psi_T 
\end{align}
where $\Psi_P^{ext}$ corresponds physically to the external poloidal flux through the hole of the solid torus and $\Psi_T = \int_{\Sigma_P} B $ is the magnetic field's flux through the solid torus (toroidal flux). This expression coincides with the usual relative helicity formula~\citep{finn-antonsen-1985,berger-1999}.

The main advantage of our derivation is that we only need to focus on the topology of the domain and how it affects the space of solutions to the boundary value problem $B=dA$. This step can be performed at the coarser homological level rather than via direct computation. Nothing needs to be assumed about the solution on the complement of the manifold (outside the domain); Hodge theory takes care of existence, uniqueness and regularity. Demonstrating these formulae in local coordinate would be extremely tedious depending on how the specific embedding affects the Riemannian metric; the coordinate-free notation is compact and general. Another advantage is that the use of so-called \emph{multi-valued} functions is avoided altogether and Stokes' theorem remains unconditionally true. There is also no need to split the domain into simply-connected components, nor to isolate the "vacuum components" of each field. Finally, the same strategy can be repeated to generalise the relative total helicity formula over to more complicated domains.

\subsection{Perfectly conducting boundary on a hollow toroidal volume}
A similar formula for the relative helicity can be derived in the special case where the magnetic field is everywhere tangential to the boundary of the hollow torus, namely when $\bm{B}\cdot \bm{n}=0$ on the exterior and interior tori. By proposition \ref{prop:ht-closedexact}, the magnetic flux (as a closed Dirichlet $2$-form $dB=0$ and $\bm{t}B=0$) is exact $B=dA$, and the total helicity of the magnetic field is well-defined by the functional of the potential $1$-form, $H[A]=A\wedge dA$. Furthermore, theorem \ref{thm:product-wedge-ht} applies and the effect of gauge can be computed by the following combination of loop and surface integrals
\begin{align}
  K[A']-K[A] = \int_{\Gamma_T}(A'-A)\int_{\Sigma_P}B - \int_{\Gamma_P}(A'-A)\int_{\Sigma_T}B,
\end{align}
where $\Gamma_T\in[\Phi(\torpath)]$, $\Gamma_P\in[\Phi(\polpath)]$ are closed toroidal and poloidal paths such that $H_1(\HT)=\text{span}\{[\Phi(\torpath)],[\Phi(\polpath)]\}$, $\Sigma_P\in[\Phi(\polocut)]$ is a poloidal cut and $\Sigma_T\in[\Phi(\toribbon)]$ a toroidal annulus (or ribbon) with $H_2(\HT,\partial \HT)=\text{span}\{[\Phi(\toribbon)],[\Phi(\polocut)]\}$. 

This suggests the following definition of \emph{relative helicity} for a perfectly conducting hollow toroidal volume
\begin{align}
  K_r[A;\Gamma_T,\Gamma_P] :&= K[A] - \int_{\Gamma_T} A\int_{\Sigma_P}B + \int_{\Gamma_P}A\int_{\Sigma_T}B.
\end{align}
This expression is gauge-invariant, $K_r[A';\Gamma_T,\Gamma_P]-K_r[A;\Gamma_T,\Gamma_P]=0$, but depends on the choice of paths $\Gamma_T$ and $\Gamma_P$. The relative helicity becomes independent of the specific path choice by restricting $\Gamma_T$ and $\Gamma_P$ to lie on the boundary of the hollow torus. The most natural choice among four possible combinations is to assign toroidal loops $\Gamma_T$ to the exterior boundary torus and poloidal loops $\Gamma_P$ to the interior. Specifically, one computes
\begin{align}
\label{eq:relative-exterior-helicity-ht}
  K^{ext}_r[A]:=K[A] -\Psi_P^{ext}\Psi_T + \Psi_P \Psi_T^{ext}
\end{align}
where $\Psi_P^{ext}$ is the exterior poloidal magnetic flux through the torus' ``major hole'', $\Psi^{ext}_T$ is the exterior toroidal magnetic flux through the ``minor hole'' (hollowed section), and $\Psi_T = \int_{\Sigma_T} B$ and $\Psi_P = \int_{\Sigma_P} B$ are respectively the toroidal and poloidal magnetic fluxes within the hollow torus. Expression (\ref{eq:relative-exterior-helicity-ht}) is a useful generalisation of the usual relative helicity formulae~\citep{finn-antonsen-1985,berger-1999} applicable to multi-region calculations of MHD equilibria~\cite{spec}.

\section{Conclusion}
Elements of Hodge theory and de Rham cohomology were introduced to formally address the effect of gauge freedom in magnetostatics problems. The Hodge-Morrey and Friedrichs decomposition theorems of $k$-forms on compact manifolds with boundary were applied to the representation problem of magnetic fluxes ($2$-forms) via potential $1$-forms on three-dimensional toroidal volumes. This coordinate-free and general method for solving boundary value problems identifies the various components in terms of boundary conditions and kernels of the exterior derivative and codifferential operator. This is a useful step to establish well-posedness and uniqueness of solutions~\cite{schwarz}.

We also highlighted the fact that, when the domain's homology is non-trivial, one needs to account for the important components belonging to the subspace of harmonic fields. This remarkable subspace is actually finite dimensional and can be characterised entirely by homological methods; the geometry (metric) does not interfere with the global (integral) properties of its constituents. The homology groups of two domains commonly used in the context of fusion plasmas were classified. All closed $2$-forms are exact on a solid toroidal volume and we saw that all physical magnetic fluxes must be exact on a hollow toroidal volume. In solid toroidal volumes, harmonic one-fields are spanned by the unique (up to sign) unit toroidal vacuum field. In axisymmetric devices (tokamaks), this $1$-form is simply a multiple of the gradient of the toroidal angle, which corresponds physically to the $1/R$ vacuum field from the toroidal field coils. In stellarators, this $1$-form is the only three-dimensional toroidal vacuum field that aligns with the boundary (last-closed flux-surface). As a vector field, the field-lines of the unit toroidal vacuum field within the domain can be extremely complex. From the perspective of stellarator design, the unit toroidal vacuum field is precisely the target configuration that is being optimised~\cite{hudson-2018,landreman-sengupta-2018,landreman-sengupta-plunk}. Questions about integral submanifolds (nested flux-surfaces), foliations, etc. naturally follow and would be important to investigate on similar grounds.

This paper provided a framework to distinguish the physically relevant potential $1$-form from the superfluous components, leading to gauge-freedom. The requirements for a minimal gauge were listed, namely that the vector potential be divergence-free, tangential and orthogonal to the space of harmonic fields. While the first condition is well-known to correspond to the Coulomb gauge, the remaining two -- specific to the domain having a boundary and a non-trivial homology -- were shown to be equally important. In fact, if any of these conditions are broken, the $L^2$-norm of the vector potential could be made arbitrary.

Theorems \ref{thm:product-wedge} and \ref{thm:product-wedge-ht} showed the remarkable splitting of the volume integral of the wedge product between a closed $1$-form and a closed Dirichlet $2$-form into a product (or a combination of products) of the line-integral of the former and the surface-integral of the latter. These results were essential in deriving relative total helicity formulae for the case of perfectly conducting boundary conditions. With these expressions, the value of the relative helicity does not depend on the choice of gauge and can be safely compared with that of another configuration. The derivation highlighted the homological origins of the well-established formulae~\cite{berger-field-1984,jensen-chu-1984,finn-antonsen-1985} and provided a formal justification for them.

The approach of this paper may help assess the effect of homology and gauge-freedom in  different applications such as the formation of eddy current formation in conducting material, magnetic levitation, Taylor-relaxed states, dynamo effect, solar flares, magnetic reconnection, etc.

\begin{acknowledgments}
  The authors would like to acknowledge stimulating discussions with S.R.Hudson, J.C.Loizu and A.Cerfon.
\end{acknowledgments}
}
\bibliographystyle{apsrev4-1}
\bibliography{biblio}

\begin{thebibliography}{24}%
\makeatletter
\providecommand \@ifxundefined [1]{%
 \@ifx{#1\undefined}
}%
\providecommand \@ifnum [1]{%
 \ifnum #1\expandafter \@firstoftwo
 \else \expandafter \@secondoftwo
 \fi
}%
\providecommand \@ifx [1]{%
 \ifx #1\expandafter \@firstoftwo
 \else \expandafter \@secondoftwo
 \fi
}%
\providecommand \natexlab [1]{#1}%
\providecommand \enquote  [1]{``#1''}%
\providecommand \bibnamefont  [1]{#1}%
\providecommand \bibfnamefont [1]{#1}%
\providecommand \citenamefont [1]{#1}%
\providecommand \href@noop [0]{\@secondoftwo}%
\providecommand \href [0]{\begingroup \@sanitize@url \@href}%
\providecommand \@href[1]{\@@startlink{#1}\@@href}%
\providecommand \@@href[1]{\endgroup#1\@@endlink}%
\providecommand \@sanitize@url [0]{\catcode `\\12\catcode `\$12\catcode
  `\&12\catcode `\#12\catcode `\^12\catcode `\_12\catcode `\%12\relax}%
\providecommand \@@startlink[1]{}%
\providecommand \@@endlink[0]{}%
\providecommand \url  [0]{\begingroup\@sanitize@url \@url }%
\providecommand \@url [1]{\endgroup\@href {#1}{\urlprefix }}%
\providecommand \urlprefix  [0]{URL }%
\providecommand \Eprint [0]{\href }%
\providecommand \doibase [0]{http://dx.doi.org/}%
\providecommand \selectlanguage [0]{\@gobble}%
\providecommand \bibinfo  [0]{\@secondoftwo}%
\providecommand \bibfield  [0]{\@secondoftwo}%
\providecommand \translation [1]{[#1]}%
\providecommand \BibitemOpen [0]{}%
\providecommand \bibitemStop [0]{}%
\providecommand \bibitemNoStop [0]{.\EOS\space}%
\providecommand \EOS [0]{\spacefactor3000\relax}%
\providecommand \BibitemShut  [1]{\csname bibitem#1\endcsname}%
\let\auto@bib@innerbib\@empty
\bibitem [{\citenamefont {Hirshman}\ and\ \citenamefont
  {Whitson}(1983)}]{vmec}%
  \BibitemOpen
  \bibfield  {author} {\bibinfo {author} {\bibfnamefont {S.~P.}\ \bibnamefont
  {Hirshman}}\ and\ \bibinfo {author} {\bibfnamefont {J.~C.}\ \bibnamefont
  {Whitson}},\ }\href@noop {} {\bibfield  {journal} {\bibinfo  {journal}
  {Physics of Fluids}\ }\textbf {\bibinfo {volume} {26}},\ \bibinfo {pages}
  {3553} (\bibinfo {year} {1983})}\BibitemShut {NoStop}%
\bibitem [{\citenamefont {Hudson}\ \emph {et~al.}(2012)\citenamefont {Hudson},
  \citenamefont {Dewar}, \citenamefont {Dennis}, \citenamefont {Hole},
  \citenamefont {McGann}, \citenamefont {von Nessi},\ and\ \citenamefont
  {Lazerson}}]{spec}%
  \BibitemOpen
  \bibfield  {author} {\bibinfo {author} {\bibfnamefont {S.~R.}\ \bibnamefont
  {Hudson}}, \bibinfo {author} {\bibfnamefont {R.~L.}\ \bibnamefont {Dewar}},
  \bibinfo {author} {\bibfnamefont {G.}~\bibnamefont {Dennis}}, \bibinfo
  {author} {\bibfnamefont {M.~J.}\ \bibnamefont {Hole}}, \bibinfo {author}
  {\bibfnamefont {M.}~\bibnamefont {McGann}}, \bibinfo {author} {\bibfnamefont
  {G.}~\bibnamefont {von Nessi}}, \ and\ \bibinfo {author} {\bibfnamefont
  {S.}~\bibnamefont {Lazerson}},\ }\href {\doibase 10.1063/1.4765691}
  {\bibfield  {journal} {\bibinfo  {journal} {Physics of Plasmas}\ }\textbf
  {\bibinfo {volume} {19}},\ \bibinfo {pages} {112502} (\bibinfo {year}
  {2012})}\BibitemShut {NoStop}%
\bibitem [{\citenamefont {Moreau}(1961)}]{moreau-1961}%
  \BibitemOpen
  \bibfield  {author} {\bibinfo {author} {\bibfnamefont {J.~J.}\ \bibnamefont
  {Moreau}},\ }\href@noop {} {\bibfield  {journal} {\bibinfo  {journal}
  {{Comptes rendus hebdomadaires des s{\'e}ances de l'Acad{\'e}mie des
  sciences}}\ }\textbf {\bibinfo {volume} {252}},\ \bibinfo {pages} {2810}
  (\bibinfo {year} {1961})}\BibitemShut {NoStop}%
\bibitem [{\citenamefont {Moffatt}(1969)}]{moffatt-1969}%
  \BibitemOpen
  \bibfield  {author} {\bibinfo {author} {\bibfnamefont {H.~K.}\ \bibnamefont
  {Moffatt}},\ }\href {\doibase 10.1017/S0022112069000991} {\bibfield
  {journal} {\bibinfo  {journal} {Journal of Fluid Mechanics}\ }\textbf
  {\bibinfo {volume} {35}},\ \bibinfo {pages} {117–129} (\bibinfo {year}
  {1969})}\BibitemShut {NoStop}%
\bibitem [{\citenamefont {Arnold}(1974)}]{arnold-1974}%
  \BibitemOpen
  \bibfield  {author} {\bibinfo {author} {\bibfnamefont {V.~I.}\ \bibnamefont
  {Arnold}},\ }in\ \href@noop {} {\emph {\bibinfo {booktitle} {Vladimir I.
  Arnold - Collected Works - Volume II}}}\ (\bibinfo  {publisher} {Springer},\
  \bibinfo {year} {1974})\ pp.\ \bibinfo {pages} {357--375}\BibitemShut
  {NoStop}%
\bibitem [{\citenamefont {Berger}\ and\ \citenamefont
  {Field}(1984)}]{berger-field-1984}%
  \BibitemOpen
  \bibfield  {author} {\bibinfo {author} {\bibfnamefont {M.~A.}\ \bibnamefont
  {Berger}}\ and\ \bibinfo {author} {\bibfnamefont {G.~B.}\ \bibnamefont
  {Field}},\ }\href {\doibase 10.1017/S0022112084002019} {\bibfield  {journal}
  {\bibinfo  {journal} {Journal of Fluid Mechanics}\ }\textbf {\bibinfo
  {volume} {147}},\ \bibinfo {pages} {133–148} (\bibinfo {year}
  {1984})}\BibitemShut {NoStop}%
\bibitem [{\citenamefont {Berger}(1999)}]{berger-1999}%
  \BibitemOpen
  \bibfield  {author} {\bibinfo {author} {\bibfnamefont {M.~A.}\ \bibnamefont
  {Berger}},\ }\href {\doibase 10.1088/0741-3335/41/12b/312} {\bibfield
  {journal} {\bibinfo  {journal} {Plasma Physics and Controlled Fusion}\
  }\textbf {\bibinfo {volume} {41}},\ \bibinfo {pages} {B167} (\bibinfo {year}
  {1999})}\BibitemShut {NoStop}%
\bibitem [{\citenamefont {Jensen}\ and\ \citenamefont
  {Chu}(1984)}]{jensen-chu-1984}%
  \BibitemOpen
  \bibfield  {author} {\bibinfo {author} {\bibfnamefont {T.~H.}\ \bibnamefont
  {Jensen}}\ and\ \bibinfo {author} {\bibfnamefont {M.~S.}\ \bibnamefont
  {Chu}},\ }\href {\doibase 10.1063/1.864602} {\bibfield  {journal} {\bibinfo
  {journal} {The Physics of Fluids}\ }\textbf {\bibinfo {volume} {27}},\
  \bibinfo {pages} {2881} (\bibinfo {year} {1984})}\BibitemShut {NoStop}%
\bibitem [{\citenamefont {Finn}\ and\ \citenamefont
  {Antonsen}(1985)}]{finn-antonsen-1985}%
  \BibitemOpen
  \bibfield  {author} {\bibinfo {author} {\bibfnamefont {J.~M.}\ \bibnamefont
  {Finn}}\ and\ \bibinfo {author} {\bibfnamefont {T.~M.~J.}\ \bibnamefont
  {Antonsen}},\ }\href@noop {} {\bibfield  {journal} {\bibinfo  {journal}
  {Comments on Plasma Physics and Controlled Fusion}\ }\textbf {\bibinfo
  {volume} {9}},\ \bibinfo {pages} {111} (\bibinfo {year} {1985})}\BibitemShut
  {NoStop}%
\bibitem [{\citenamefont {Taylor}(1986)}]{taylor-1986}%
  \BibitemOpen
  \bibfield  {author} {\bibinfo {author} {\bibfnamefont {J.~B.}\ \bibnamefont
  {Taylor}},\ }\href {\doibase 10.1103/RevModPhys.58.741} {\bibfield  {journal}
  {\bibinfo  {journal} {Rev. Mod. Phys.}\ }\textbf {\bibinfo {volume} {58}},\
  \bibinfo {pages} {741} (\bibinfo {year} {1986})}\BibitemShut {NoStop}%
\bibitem [{\citenamefont {de~Rham}(1931)}]{derham-1931}%
  \BibitemOpen
  \bibfield  {author} {\bibinfo {author} {\bibfnamefont {G.}~\bibnamefont
  {de~Rham}},\ }\href@noop {} {\bibfield  {journal} {\bibinfo  {journal}
  {Journal de Mathématiques Pures et Appliquées}\ }\textbf {\bibinfo {volume}
  {10}},\ \bibinfo {pages} {115} (\bibinfo {year} {1931})}\BibitemShut
  {NoStop}%
\bibitem [{\citenamefont {Hodge}(1941)}]{hodge-1941}%
  \BibitemOpen
  \bibfield  {author} {\bibinfo {author} {\bibfnamefont {W.}~\bibnamefont
  {Hodge}},\ }\href@noop {} {\emph {\bibinfo {title} {The Theory and
  Applications of Harmonic Integrals}}}\ (\bibinfo  {publisher} {University
  Press},\ \bibinfo {year} {1941})\BibitemShut {NoStop}%
\bibitem [{\citenamefont {Holm}\ \emph {et~al.}(1998)\citenamefont {Holm},
  \citenamefont {Marsden},\ and\ \citenamefont {Ratiu}}]{holm-marsden-ratiu}%
  \BibitemOpen
  \bibfield  {author} {\bibinfo {author} {\bibfnamefont {D.~D.}\ \bibnamefont
  {Holm}}, \bibinfo {author} {\bibfnamefont {J.~E.}\ \bibnamefont {Marsden}}, \
  and\ \bibinfo {author} {\bibfnamefont {T.~S.}\ \bibnamefont {Ratiu}},\ }\href
  {\doibase https://doi.org/10.1006/aima.1998.1721} {\bibfield  {journal}
  {\bibinfo  {journal} {Advances in Mathematics}\ }\textbf {\bibinfo {volume}
  {137}},\ \bibinfo {pages} {1 } (\bibinfo {year} {1998})}\BibitemShut
  {NoStop}%
\bibitem [{\citenamefont {Ono}(1995)}]{ono-1995}%
  \BibitemOpen
  \bibfield  {author} {\bibinfo {author} {\bibfnamefont {T.}~\bibnamefont
  {Ono}},\ }\href {\doibase https://doi.org/10.1016/0167-2789(94)00152-G}
  {\bibfield  {journal} {\bibinfo  {journal} {Physica D: Nonlinear Phenomena}\
  }\textbf {\bibinfo {volume} {81}},\ \bibinfo {pages} {207 } (\bibinfo {year}
  {1995})}\BibitemShut {NoStop}%
\bibitem [{\citenamefont {Hattori}(1994)}]{hattori-1994}%
  \BibitemOpen
  \bibfield  {author} {\bibinfo {author} {\bibfnamefont {Y.}~\bibnamefont
  {Hattori}},\ }\href {\doibase 10.1088/0305-4470/27/2/004} {\bibfield
  {journal} {\bibinfo  {journal} {Journal of Physics A: Mathematical and
  General}\ }\textbf {\bibinfo {volume} {27}},\ \bibinfo {pages} {L21}
  (\bibinfo {year} {1994})}\BibitemShut {NoStop}%
\bibitem [{\citenamefont {Taylor}(1974)}]{taylor-1974}%
  \BibitemOpen
  \bibfield  {author} {\bibinfo {author} {\bibfnamefont {J.~B.}\ \bibnamefont
  {Taylor}},\ }\href {\doibase 10.1103/PhysRevLett.33.1139} {\bibfield
  {journal} {\bibinfo  {journal} {Phys. Rev. Lett.}\ }\textbf {\bibinfo
  {volume} {33}},\ \bibinfo {pages} {1139} (\bibinfo {year}
  {1974})}\BibitemShut {NoStop}%
\bibitem [{\citenamefont {Schwarz}(1995)}]{schwarz}%
  \BibitemOpen
  \bibfield  {author} {\bibinfo {author} {\bibfnamefont {G.}~\bibnamefont
  {Schwarz}},\ }\href@noop {} {\emph {\bibinfo {title} {Hodge decomposition: A
  method for solving boundary value problems}}},\ Lecture notes in mathematics\
  (\bibinfo  {publisher} {Springer},\ \bibinfo {year} {1995})\BibitemShut
  {NoStop}%
\bibitem [{\citenamefont {Hudson}\ \emph {et~al.}(2018)\citenamefont {Hudson},
  \citenamefont {Zhu}, \citenamefont {Pfefferlé},\ and\ \citenamefont
  {Gunderson}}]{hudson-2018}%
  \BibitemOpen
  \bibfield  {author} {\bibinfo {author} {\bibfnamefont {S.}~\bibnamefont
  {Hudson}}, \bibinfo {author} {\bibfnamefont {C.}~\bibnamefont {Zhu}},
  \bibinfo {author} {\bibfnamefont {D.}~\bibnamefont {Pfefferlé}}, \ and\
  \bibinfo {author} {\bibfnamefont {L.}~\bibnamefont {Gunderson}},\ }\href
  {\doibase https://doi.org/10.1016/j.physleta.2018.07.016} {\bibfield
  {journal} {\bibinfo  {journal} {Physics Letters A}\ }\textbf {\bibinfo
  {volume} {382}},\ \bibinfo {pages} {2732 } (\bibinfo {year}
  {2018})}\BibitemShut {NoStop}%
\bibitem [{\citenamefont {Landreman}\ and\ \citenamefont
  {Sengupta}(2018)}]{landreman-sengupta-2018}%
  \BibitemOpen
  \bibfield  {author} {\bibinfo {author} {\bibfnamefont {M.}~\bibnamefont
  {Landreman}}\ and\ \bibinfo {author} {\bibfnamefont {W.}~\bibnamefont
  {Sengupta}},\ }\href {\doibase 10.1017/S0022377818001289} {\bibfield
  {journal} {\bibinfo  {journal} {Journal of Plasma Physics}\ }\textbf
  {\bibinfo {volume} {84}},\ \bibinfo {pages} {905840616} (\bibinfo {year}
  {2018})}\BibitemShut {NoStop}%
\bibitem [{\citenamefont {Landreman}\ \emph {et~al.}(2019)\citenamefont
  {Landreman}, \citenamefont {Sengupta},\ and\ \citenamefont
  {Plunk}}]{landreman-sengupta-plunk}%
  \BibitemOpen
  \bibfield  {author} {\bibinfo {author} {\bibfnamefont {M.}~\bibnamefont
  {Landreman}}, \bibinfo {author} {\bibfnamefont {W.}~\bibnamefont {Sengupta}},
  \ and\ \bibinfo {author} {\bibfnamefont {G.~G.}\ \bibnamefont {Plunk}},\
  }\href {\doibase 10.1017/S0022377818001344} {\bibfield  {journal} {\bibinfo
  {journal} {Journal of Plasma Physics}\ }\textbf {\bibinfo {volume} {85}},\
  \bibinfo {pages} {905850103} (\bibinfo {year} {2019})}\BibitemShut {NoStop}%
\bibitem [{\citenamefont {Brown}(1962)}]{brown-1962}%
  \BibitemOpen
  \bibfield  {author} {\bibinfo {author} {\bibfnamefont {M.}~\bibnamefont
  {Brown}},\ }\href {http://www.jstor.org/stable/1970177} {\bibfield  {journal}
  {\bibinfo  {journal} {Annals of Mathematics}\ }\textbf {\bibinfo {volume}
  {75}},\ \bibinfo {pages} {331} (\bibinfo {year} {1962})}\BibitemShut
  {NoStop}%
\bibitem [{\citenamefont {Helmoltz}(1858)}]{helmoltz-1858}%
  \BibitemOpen
  \bibfield  {author} {\bibinfo {author} {\bibfnamefont {H.}~\bibnamefont
  {Helmoltz}},\ }\href {\doibase 10.1515/crll.1858.55.25} {\bibfield  {journal}
  {\bibinfo  {journal} {Journal für die reine und angewandte Mathematik}\ }
  (\bibinfo {year} {1858}),\ 10.1515/crll.1858.55.25}\BibitemShut {NoStop}%
\bibitem [{\citenamefont {Morrey}(1956)}]{morrey-1956}%
  \BibitemOpen
  \bibfield  {author} {\bibinfo {author} {\bibfnamefont {C.~B.}\ \bibnamefont
  {Morrey}},\ }\href {http://www.jstor.org/stable/2372488} {\bibfield
  {journal} {\bibinfo  {journal} {American Journal of Mathematics}\ }\textbf
  {\bibinfo {volume} {78}},\ \bibinfo {pages} {137} (\bibinfo {year}
  {1956})}\BibitemShut {NoStop}%
\bibitem [{\citenamefont {Friedrichs}(1955)}]{friedrichs-1955}%
  \BibitemOpen
  \bibfield  {author} {\bibinfo {author} {\bibfnamefont {K.~O.}\ \bibnamefont
  {Friedrichs}},\ }\href {\doibase 10.1002/cpa.3160080408} {\bibfield
  {journal} {\bibinfo  {journal} {Communications on Pure and Applied
  Mathematics}\ }\textbf {\bibinfo {volume} {8}},\ \bibinfo {pages} {551}
  (\bibinfo {year} {1955})}\BibitemShut {NoStop}%
\end{thebibliography}%

\appendix

\section{Hodge theory}
\label{sec:hodge-theory}
{\nochange 
We will be adopting the notation from the book by~\citet{schwarz} and reporting the most useful results for our purposes.

In order to avoid confusion, we reserve the symbol $H$ to refer to homology or cohomology classes $H_k(X)$ and $H^k(X)$. Sobolev-Hilbert spaces will be denoted by $W^1(X)=W^{1,2}(X)$ instead.
}

\subsection{The Hilbert spaces of square-integrable forms}

{\nochange

The vector space $\Omega^k(M)$ of smooth $k$-forms on an oriented compact smooth Riemannian $n$-manifold $(M,\langle\cdot,\cdot\rangle)$ with boundary is equipped with the $L^2$-inner product, $\bra\cdot,\cdot\ket:\Omega^k(M)\times\Omega^k(M)\mapsto \R$ defined by
\begin{align}
  \bra \alpha,\beta\ket := \int_M  \alpha \wedge \star \beta = \int_M \langle \alpha,\beta\rangle \mu
\end{align}
where $\mu\in \Omega^n(M)$ is the natural volume-form of the Riemannian metric $\langle\cdot,\cdot\rangle$ on $M$, and $\star:\Omega^k(M)\to \Omega^{n-k}(M)$ is the Hodge star operator. Let $L^2\Omega^k(M)$ be the $L^2$-completion of $\Omega^k(M)$. Denote by $W^1\Omega^k(M)$ the completion of $\Omega^k(M)$ with respect to the so-called Dirichlet inner product $\bra \alpha,\beta\ket_{\mathcal{D}}:=\bra d\alpha,d\beta\ket + \bra \delta \alpha,\delta \beta\ket$, where $d$ is the (metric-independent) exterior derivative and $\delta:=(-1)^k\star^{-1}d\star$ is the (metric-dependent) codifferential.

On manifolds with boundary, it is useful to label the subspaces of $k$-forms respecting Dirichlet and Neumann boundary conditions. Let $j:\partial M\to M$ be the natural inclusion, with tangent map $Tj:T\partial M\to TM|_{\partial M}$. The tangent bundle $T\partial M$ of the boundary manifold can be identified with the image of the tangent map, $Tj(\partial M)$, but not with the restriction of the tangent bundle to the boundary $TM|_{\partial M}$. We construct unit normal vector fields $\mathcal{N}\in \Gamma(TM|_{\partial M})$ such that $\langle \mathcal{N},\mathcal{N}\rangle = 1$ and $\langle \mathcal{N}, Tj (Y)\rangle = 0$, $\forall Y\in\Gamma(T\partial M)$, and extend them in a neighbourhood of the boundary by the collar theorem~\cite{brown-1962}. Any vector field $X\in\Gamma(TM)$ can then be decomposed on the boundary into tangential and normal components as
\begin{align*}
  X^\perp &= \langle X|_{\partial M},\mathcal{N}\rangle \mathcal{N}, &
  X^{||} &= X|_{\partial M} - X^\perp.
\end{align*}
In turn, the tangential and normal boundary operators are defined on $k$-forms $\alpha\in\Omega^k(M)$ by
\begin{align*}
  \bm{n}\alpha(X_1,\ldots,X_k) &:= \alpha(X_1^\perp,\ldots,X_k^\perp),&
  \bm{t}\alpha := \alpha|_{\partial M} - \bm{n}\alpha.
\end{align*}
for any set of vector fields on the boundary $X_1,\ldots,X_k\in \Gamma(TM|_{\partial M})$. Tangential projection coincides with the pullback of the inclusion in the sense that
\begin{align*}
  j^*\bm{t}\alpha &= j^*\alpha , &
  j^*\bm{n}\alpha &= 0.
\end{align*}
Subscripts $D$ and $N$ denote the subspaces of Dirichlet $k$-forms and Neumann $k$-forms:
\begin{align}
  \begin{split}
  \Omega_D^k(M):&=\{ \alpha \in \Omega^k(M)|\ \bm{t}\alpha=0\}, \\
  \Omega_N^k(M):&=\{ \alpha \in \Omega^k(M)|\ \bm{n}\alpha=0\}.
  \end{split}
\end{align}
The tangential and normal operators satisfy the following (commutation) relations (see~\citet[Proposition 1.2.6]{schwarz})
\begin{align}
  \bm{t} d\omega &=d \bm{t} \omega,&
  \bm{n}\delta\omega &=\delta\bm{n}\omega,\\
  \star \bm{n}\omega &= \bm{t}\star \omega, &
  \star \bm{t}\omega &= \bm{n}\star \omega.\label{eq:star-bound-op}
\end{align}

Stokes theorem is the fundamental result of exterior calculus on which Hodge decomposition of $k$-forms is based. For any $n-1$-form $\omega \in W^1\Omega^{n-1}(M)$, one has
\begin{align}
  \int_M d\omega = \int_{\partial M} \bm{t}\omega.
\end{align}
Importantly, the integral over a closed manifold $M$ (whose boundary is empty $\partial M=\varnothing$) of an exact form $d\omega$ vanishes. Exact forms are always closed, since $dd=0$ but the converse is not always true. Non-exact closed forms represent so-called \emph{cohomology classes} that are nonzero. Homotopy equivalences (in particular homeomorphisms) respect cohomology classes.

Green's formula is a useful corollary to Stokes theorem, demonstrating that the exterior derivative and the codifferential are almost dual operations under the $L^2$-pairing (see~\citet[Proposition 2.1.2]{schwarz}):
\begin{align}
  \label{eq:green}
  \bra d\omega,\eta\ket = \bra \omega,\delta \eta\ket
  + \int_{\partial M} \bm{t}\omega \wedge \star \bm{n}\eta.
\end{align}
Duality is achieved by restricting to families of forms that satisfy suitable boundary conditions, thereby eliminating the boundary term in equation (\ref{eq:green}). This is an important ingredient of the decomposition theorems below.
}

\subsection{de Rham cohomology in a nutshell}

{\nochange
Two forms, $s_1,s_2\in W^1\Omega^k(M)$ that are closed, namely $ds_1=ds_2=0$, are said to belong to the same cohomology class (are cohomologous) when their difference is exact, namely $s_1-s_2=d\epsilon$ for some $\epsilon \in W^1\Omega^{k-1}(M)$. The $k^{th}$ de Rham cohomology group of $M$ is the quotient (vector) space
\begin{align}
  H_{dR}^k(M,d):=\ker d|_{W^1\Omega^k(M)}/ \im d|_{W^1\Omega^{k-1}(M)},
\end{align}
namely the real vector space of closed $k$-forms modulo the space of exact $k$-forms.


Integration of $k$-forms over $k$-chains produces, by virtue of Stokes theorem, a well-defined (homomorphism) bilinear pairing $\ll\cdot,\cdot\gg : H_k(M)\times H_{dR}^k(M,d)\to \R$,
\begin{align}
  \ll [\Gamma],[\lambda] \gg\ := \int_\Gamma\lambda,
\end{align}
between $H_k(M)$, the $k^{th}$ singular homology group (space of $k$-cycles modulo the $k$-boundaries) and the $k^{th}$ de Rham cohomology group. 

De Rham's theorem states that the map induced by this pairing is an isomorphism from $H_{dR}^k(M,d)$ to the dual $H^k(M,\R)$ of $H_k(M,\R)$, namely $H_{dR}^k(M,d)\cong H^k(M,\R)$. Nonzero elements of de Rham cohomology groups are equivalence classes represented by closed forms whose integrals over closed $k$-cycles do not all vanish. So another way of defining the de Rham cohomology group is
\begin{align*}
  H_{dR}^k(M,d):=\ker d|_{W^1\Omega^k(M)}/ \ker \int_{\Gamma^k}, \quad \forall \Gamma^k\in H_k(M).
\end{align*}
On a compact manifold, the dimension of $H_k(M,\R)$ is a finite non-negative integer called the $k^{th}$ Betti number. Then, 
given a basis $\{[\Gamma_i]\}_{i=1}^{\beta_k}$ of $H_k(M,\R)$, there is a dual basis 
 $\{[\lambda_i]\}_{i=1}^{\beta_k}$, whose representatives are closed $k$-forms, for which $\int_{\Gamma_i}\lambda_j$ is Kronecker's
$\delta_{ij}$. This makes no use of a Riemannian metric on the manifold $M$.}

A similar isomorphism can be established over the \emph{relative} cohomology with respect to the manifold's boundary, namely for $k$-forms satisfying the Dirichlet boundary condition, 
\begin{align}
    H_{dR}^k(M,\partial M,d)
    :=\ker d|_{W^1\Omega^k_D(M)} / \im d|_{W^1\Omega^{k-1}_D(M)},
\end{align}
and when chains on the boundary are quotiented out. Indeed, the non-degenerate bilinear pairing $\ll \cdot,\cdot \gg: H_k(M,\partial M)\times H_{dR}^k(M,\partial M,d)\to \R$ with the same rule as above is perfect, so that $H_{dR}^k(M,\partial M,d)\cong H^k(M,\partial M,\R)$. 

\subsection{$L^2$-decomposition theorems on manifolds with boundary}
\label{sec:hfmd}
{\nochange
Now, making use of a Riemannian metric on $M$, we describe the Hodge-Friedrichs-Morrey decomposition. HFMD is a generalisation of the classical Helmoltz decomposition theorem~\citep{helmoltz-1858} in $\R^3$, and is a remarkable result on the splitting of the infinite-dimensional Hilbert space $L^2\Omega^k(M)$ into a sum of orthogonal subspaces. HFMD is actually the result of the Hodge-Morrey decomposition~\citep{morrey-1956} (HMD) together with the Friedrichs decomposition~\citep{friedrichs-1955} (FD), which we briefly discuss.

The HMD theorem says that $L^2\Omega^k(M)$ is the direct-sum of exact Dirichlet k-forms, co-exact Neumann k-forms and so-called harmonic k-fields (no boundary conditions). Reciting \citet[Theorem 2.4.2]{schwarz}, the conclusion of the HMD is that the Hilbert space of square integrable $k$-forms on a compact orientable manifold with boundary splits into the direct sum
\begin{align}
  \label{eq:hmd}
  L^2\Omega^k(M) = \mathcal{E}^k(M) \oplus \mathcal{C}^k(M)\oplus L^2\Harm^k(M)
\end{align}
where 
\begin{align}
  \mathcal{E}^k(M) :&= \im d\big|_{W^1\Omega_D^{k-1}(M)} \subset L^2\Omega_D^k(M)\nonumber\\
  &= \{ d\alpha |\ \alpha \in W^1\Omega^{k-1}(M), \bm{t}\alpha = 0 \}
\end{align}
is the space of exact $k$-forms produced by the exterior derivative of any $k-1$-form with vanishing tangential component ($\mathcal{E}^0(M):=\{0\}$),
\begin{align}
  \mathcal{C}^k(M) :&=\im \delta\big|_{W^1\Omega_N^{k+1}(M)}\subset L^2\Omega_N^k(M)\nonumber\\
  &=\{ \delta\beta|\ \beta \in W^1\Omega^{k+1}(M),\ \bm{n}\beta = 0 \}
\end{align}
is the space of co-exact $k$-forms produced by the codifferential of any $k+1$ form with vanishing normal component ($\mathcal{C}^n(M):=\{0\}$),
\begin{align}
  \Harm^k(M) :&=\ker d|_{W^1\Omega^k(M)}\cap \ker \delta \big|_{W^1\Omega^k(M)}\nonumber\\
  &=\{ \kappa \in W^1\Omega^k(M) |\ d\kappa=0,\ \delta \kappa=0 \}
\end{align}
is the space of \emph{harmonic} $k$-fields. Notice that being a harmonic $k$-field is more restrictive than being a harmonic $k$-form; $d\lambda=0$ and $\delta \lambda=0$ implies that $\Delta \lambda = (d\delta+\delta d)\lambda = 0 $, but the converse is in general not true on manifolds with non-empty boundary. It is worth mentioning that the spaces $\mathcal{E}^k(M)$ and $\mathcal{C}^k(M)$ are closed in the $L^2$-topology (see~\citet[Lemma 2.4.3]{schwarz}). 

The FD theorem says that the space of harmonic fields on a compact manifold with boundary can be further split into two orthogonal subspaces. This can be done in two ways (\citet[Theorem 2.4.8]{schwarz}),
\begin{align}
  \label{eq:fd}
  \begin{split}
  L^2\Harm^k(M) &= \Harm_D^k(M)\oplus L^2\Harm_{co}^k(M) \\
  &= \Harm_N^k(M) \oplus L^2\Harm_{ex}^k(M)
\end{split}
\end{align}
where
\begin{multline}
  \Harm_D^k(M)
  :=\{ \lambda_D \in W^1\Omega^k(M) |\\
  \ d\lambda_D=0, \delta \lambda_D=0, \bm{t}\lambda_D=0\}
\end{multline}
is the space of Dirichlet harmonic fields,
\begin{multline}
  \Harm_N^k(M)
  :=\{ \lambda_N \in W^1\Omega^k(M) |\\
  \ d\lambda_N=0, \delta \lambda_N=0, \bm{n}\lambda_N=0\}
\end{multline}
is the space of Neumann harmonic fields, 
\begin{align}
  \Harm_{co}^k(M) := \{ \delta \gamma |\ \gamma \in W^1\Omega^{k+1}(M), d\delta\gamma = 0 \}
\end{align}
is the space of co-exact harmonic fields, and
\begin{align}
  \Harm_{ex}^k(M) := \{ d \epsilon |\ \epsilon \in W^1\Omega^{k-1}(M), \delta d\epsilon = 0 \}
\end{align}
is the space of exact harmonic fields. Dirichlet and Neumann harmonic fields are smooth (see~\citet[Theorem 2.2.6 and 2.2.7]{schwarz}).

\subsection{General consequences and useful properties of Hodge theory}
To illustrate the use of Hodge theory for solving boundary value problems, we list some useful consequences of this decomposition in the form of propositions with proofs.

\begin{proposition}
  \label{prop:hmd-explicit}
  Any $k$-form $\omega\in L^2\Omega^k(M)$ can be uniquely (here and elsewhere) written as
\begin{align*}
  \omega = d\alpha + \delta \beta + \kappa
\end{align*}
where $\alpha\in W^1\Omega^{k-1}_D(M)$, $\beta\in W^1\Omega^{k+1}_N(M)$ and $\kappa\in L^2\Harm^k(M)$ is closed and co-closed. The three terms $d\alpha$, $\delta\beta$ and $\kappa$ are mutually orthogonal with respect to the $L^2$-inner product.
\end{proposition}
The freedom in choosing $\alpha$ and $\beta$, in their respective spaces, corresponds to \emph{gauge-freedom} in the wider context of $k$-forms.

\begin{proof}
  This is a direct application of the HMD, equation~(\ref{eq:hmd}). Orthogonality is shown via Green's formula,
  \begin{align*}
  \bra d\alpha ,\kappa\ket = \bra \alpha,\cancel{\delta\kappa}\ket
  +\int_{\partial M} \cancel{\bm{t}\alpha}\wedge\star\bm{n}\kappa
  = 0,
\end{align*}
and, similarly, $\bra d\alpha,\delta \beta \ket = 0$ and $\bra\delta \beta,\kappa\ket=0$.
\end{proof}

\begin{proposition}
  \label{prop:l2-norm-decomp} Using the notation from the previous result, the $L^2$-norm of any $k$-form $\omega$ is computed as
\begin{align*}
  ||\omega||_{L^2}^2 = \bra d\alpha, d\alpha\ket + \bra \delta \beta,\delta \beta\ket + \bra \kappa,\kappa\ket,
\end{align*}
\end{proposition}

\begin{proof}
$\bra \omega,\omega\ket
  = \bra d\alpha+\delta \beta+\kappa,d\alpha+\delta \beta+\kappa\ket
  = \bra d\alpha,d\alpha \ket+\bra \delta \beta,\delta \beta\ket + \bra \kappa,\kappa\ket
  + 2\cancel{[\bra d\alpha,\delta\beta\ket + \bra d\alpha,\kappa\ket + \bra \delta\beta,\kappa\ket]}$.
\end{proof}

\begin{proposition}
  \label{prop:closed-exact}
  The following pair of equivalences hold:
  \begin{enumerate}[wide]
  \item A $k$-form $\omega\in L^2\Omega^k(M)$ is closed $d\omega = 0 \iff \bra \omega, \delta \beta\ket = 0, \forall \beta\in W^1\Omega_N^{k+1}(M)$, and thus $\omega = d\alpha +\kappa$ where $\alpha\in W^1\Omega_D^{k-1}(M)$ and $\mathcal{\kappa}\in L^2\Harm^k(M)$.
  \item A $k$-form $\tilde{\omega}\in L^2\Omega^k(M)$ is co-closed $\delta\tilde{\omega} = 0 \iff \bra \omega, d \alpha \ket = 0, \forall \alpha\in \Omega_D^{k-1}(M)$, and thus $\omega = \delta\beta +\tilde{\kappa}$ where $\beta\in W^1\Omega_N^{k+1}(M)$ and $\tilde{\kappa}\in L^2\Harm^k(M)$.
  \end{enumerate}
\end{proposition}

\begin{proof}
  By Green's formula, $\forall \alpha\in \Omega_D^{k-1}(M)$ and $\forall \beta\in \Omega_N^{k+1}(M)$,
  \begin{align*}
    \bra \alpha,\delta\omega\ket &= \bra d \alpha,\omega\ket + \int_{\partial M} \cancel{\bm{t}\alpha}\wedge\star\bm{n}\omega = \bra \omega, d \alpha\ket\\
\bra d\omega,\beta\ket &= \bra \omega,\delta\beta\ket + \int_{\partial M} \bm{t}\omega\wedge\star\cancel{\bm{n}\beta} = \bra \omega, \delta \beta\ket.
\end{align*}
\end{proof}

The following is an explicit reading of~\citet[Corollary 2.4.9]{schwarz}.
\begin{proposition}[Helmoltz decomposition]
  \label{prop:helmoltz-decomp}
  Any $k$-form $\omega\in L^2\Omega^k(M)$ can be orthogonally decomposed into an exact $k$-form and a co-closed Neumann $k$-form,
  \begin{align*}
    \omega&= d\eta + \rho,&
    \eta&\in W^1\Omega^{k-1}(M),&
    \rho&\in W^1\Omega_N^k(M), \delta\rho=0
  \end{align*}
  with $\bra d\eta,\rho\ket=0$.
  
  By duality, that same $k$-form can be orthogonally decomposed into a co-exact $k$-form and a closed Dirichlet $k$-form,
  \begin{align*}
    \omega&=\delta \sigma + \tau,&
    \sigma&\in W^1\Omega^{k+1}(M), &
    \tau&\in W^1\Omega_D^k(M), d\tau=0
  \end{align*}
  with $\bra \delta\sigma,\tau\ket=0$.
\end{proposition}

\begin{proof}
  By HMD, $\omega = d\alpha+\delta \beta + \kappa$ as in Proposition \ref{prop:hmd-explicit}.
  By $FD$, the harmonic $k$-field can be uniquely expressed as
  \begin{align*}
    \kappa=d\epsilon + \lambda_N \Rightarrow
    \omega = d(\underbrace{\alpha+\epsilon}_{\eta})+\underbrace{\delta\beta+\lambda_N}_{\rho}
  \end{align*}
  where $\lambda_N\in \Harm_N^k(M)$ and $\epsilon \in W^1\Omega^{k-1}(M)$ is such that $\delta d\epsilon=0$. Alternatively,
  \begin{align*}
    \kappa=\delta \gamma + \lambda_D \Rightarrow
    \omega = \delta(\underbrace{\beta + \gamma}_{\sigma}) + \underbrace{d\alpha+\lambda_D}_{\tau}
  \end{align*}
  where $\lambda_D\in \Harm_D^k(M)$ and $\gamma\in W^1\Omega^{k+1}(M)$ is such that $d\delta\gamma=0$.
\end{proof}
}

\begin{proposition}
\label{prop:hodge-duality}
  The Hodge star operator on $\Omega^k(M)$ induces the isomorphism $\Harm_N^k(M)\cong \Harm_D^{n-k}(M)$.
\end{proposition}

\begin{proof}
   Let $\lambda_N\in \Harm_N^k(M)$. Then, $\delta \star \lambda_N=(-1)^k \star^{-1} d\star^2 \lambda_N = (-1)^{k(n-k+1)}\star^{-1}d\lambda_N = 0$, %
   $d\star \lambda_N =
   (-1)^k \star \delta \lambda_N=0$, and by (\ref{eq:star-bound-op}), $\bm{n}\star \lambda_N = \star \bm{t}\lambda_N = 0$. Thus, $\star( \Harm_N^k(M))\subset \Harm_D^{n-k}(M)$. Similarly, one shows that $\star^{-1}\Harm_D^{n-k}(M)\subset \Harm_N^k(M)$.
\end{proof}

The following two theorems coincide with~\citet[Theorem 2.6.1]{schwarz}.
\begin{theorem}[Hodge isomorphism]
\label{thm:hodge-isomorphism}
$\Harm_N^k(M) \cong H^k_{dR}(M,d)$ and  $\Harm_D^k(M)\cong H^k_{dR}(M,\delta)$.
\end{theorem}

\begin{proof}
  By proposition \ref{prop:closed-exact} and FD, every closed $k$-form is uniquely decomposed as $\omega = d\alpha + d\epsilon + \lambda_N$ where $\alpha\in W^1\Omega_D^{k-1}(M)$, $\epsilon\in W^1\Omega^{k-1}(M)$ is such that $\delta d\epsilon=0$ and $\lambda_N\in \Harm_N^k(M)$. Thus, every closed $k$-form is cohomologous to an element of $\Harm_N^k(M)$ and, by orthogonality of the decomposition, $\Harm_N^k(M)\cong \ker d|_{W^1\Omega^k(M)}/ \im d|_{W^1\Omega^{k-1}(M)}$.
  
  Similarly, every co-closed $k$-form is decomposed as $\tilde{\omega}=\delta \beta + \delta \gamma + \lambda_D$ where $\beta\in W^1\Omega^{k+1}(M)$, $\gamma\in W^1\Omega^{k+1}(M)$ is such that $d\delta\gamma=0$ and $\lambda_D\in \Harm_D^k(M)$. Hence, $\Harm_D^k(M)\cong \ker \delta|_{W^1\Omega^k(M)}/ \im \delta|_{W^1\Omega^{k+1}(M)}$.
\end{proof}

\begin{theorem}
\label{thm:hodge-isomorphism2}
  $\Harm_D^k(M)\cong H_{dR}^k(M,\partial M ,d)$ and $\Harm_N^k(M)\cong H_{dR}^k(M,\partial M,\delta)$.
\end{theorem}

\begin{proof}
By proposition \ref{prop:closed-exact} and FD, every closed $k$-form is uniquely decomposed as $\tilde{\omega} = d\alpha + \delta \gamma + \lambda_D$, where $\alpha\in W^1\Omega_D^{k-1}(M)$, $\gamma\in W^1\Omega^{k+1}(M)$ is such that $d\delta \gamma=0$ and $\lambda_D\in \Harm_D^k(M)$. If $\tilde{\omega}$ is Dirichlet, $0=\bm{t}\tilde{\omega}=\bm{t}\delta \gamma$, then by Green's formula
\begin{align*}
    ||\delta \gamma||_{L^2} = \bra \delta \gamma,\delta \gamma \ket
    = \bra \cancel{d\delta\gamma},\gamma \ket - \int_{\partial M} \cancel{\bm{t}\delta} \gamma\wedge \star \bm{n} \gamma 
    =0,
\end{align*}
which implies that $\delta \gamma=0$ and $\tilde{\omega}=d\alpha+\lambda_D$. Hence, every closed Dirichlet $k$-form is cohomologous to an element of $\Harm^k_D(M)$ and, by orthogonality of the decomposition, $\Harm_D^k(M)\cong \ker d|_{W^1\Omega_D^k(M)}/ \im d|_{W^1\Omega_D^{k-1}(M)}$. The dual isomorphism is shown correspondingly.
\end{proof}

The following coincides with~\citet[Corollary 2.6.2]{schwarz}.
\begin{corollary}[Poincaré-Lefschetz duality]
\label{thm:poincare-duality}
The following isomorphisms hold:
\begin{align*}
    H_{dR}^k(M,\partial M,d)&\cong H_{dR}^k(M,\delta),\\
    H_{dR}^k(M,d)&\cong H_{dR}^k(M,\partial M,\delta),\\
    H_{dR}^k(M,d)&\cong H_{dR}^{n-k}(M,\delta),\\
    H_{dR}^k(M,\partial M,d)&\cong H_{dR}^{n-k}(M,\partial M,\delta),\\
    H_{dR}^k(M,d)&\cong H_{dR}^{n-k}(M,\partial M,d),\\
    H_{dR}^k(M,\delta)&\cong H_{dR}^{n-k}(M,\partial M,\delta).
\end{align*}
\end{corollary}

It is a remarkable result that Dirichlet and Neumann harmonic fields can be identified with de Rham cohomology classes. Every de Rham cohomology class has a unique Neumann harmonic representative, and a unique Dirichlet harmonic representative. The dimension of the spaces of Dirichlet and Neumann harmonic fields is thus finite and can be computed through homological methods; $\text{dim} \ \Harm_N^k(M) = \beta_k$ the $k$-th Betti number, and $\text{dim}\ \Harm_D^k(M) =\beta_{n-k}$.

The only harmonic $k$-field that simultaneously satisfies Dirichlet and Neumann boundary conditions is the zero-form, $\Harm^k_D(M)\cap\Harm^k_N(M)=\{ 0 \}$. For $k=0$ and $k=n$, respectively, harmonic fields are constants, i.e. if $M$ is connected, then $\Harm^0(M)=\Harm^0_N(M)=\R$ ($\Harm^0_D(M)=\{0\}$) and $\Harm^n(M)=\{ c\mu \in \Omega^n(M)|\, c\in \R\}=\Harm^n_D(M)$ where $\mu$ is the Riemannian volume-form ($\Harm^n_N(M)=\{0\}$).

\end{document}